\documentclass[a4paper,12pt,table]{article}
\usepackage{float}
\usepackage{amsfonts}
\usepackage{appendix}
\usepackage{amsmath}
\usepackage{amssymb}
\usepackage{amsthm}
\usepackage[longnamesfirst]{natbib}
\usepackage[utf8]{inputenc} 
\usepackage{color}
\usepackage{setspace,vmargin}
\usepackage{hyperref,url}

\setmarginsrb{20mm}{10mm}{20mm}{20mm}{10mm}{10mm}{10mm}{10mm}

\usepackage{pstricks}
\usepackage{xcolor,colortbl}

\newtheorem{theorem}{Theorem}[section]

\newtheorem{lemma}{Lemma}[section]

\newtheorem{assumption}{Assumption}[section]

\theoremstyle{definition}

\newtheorem{remark}{Remark}

\usepackage
{todonotes}

\definecolor{lightblue}{rgb}{0,0.9,1}

\makeatletter  
\def\@endtheorem{\hfill $\diamond$\endtrivlist\@endpefalse } 
\makeatother

\definecolor{MyRed}{rgb}{0.6,0,0}
\definecolor{MyGreen}{rgb}{0,0.5,0.1}
\definecolor{MyBlue}{rgb}{0,0,0.6}
\definecolor{MyGrey}{rgb}{0.3,0.3,0.4}

\hypersetup{colorlinks,
linkcolor=black,citecolor=MyGrey,filecolor=MyBlue, pdfborder={0 0 0},anchorcolor=MyBlue, urlcolor=MyBlue,pdfstartview=FitH}

\makeatletter
\renewcommand\paragraph{\@startsection{paragraph}{4}{\z@}%
                                    {3.25ex \@plus1ex \@minus.2ex}%
                                    {-1em}%
                                    {\normalfont\normalsize\bfseries\color{MyBlue}}}
\makeatother

\renewcommand{\emph}[1]{{\em{\textcolor{MyBlue}{#1}}}}

\usepackage{tabularx,listings}
\newsavebox{\codebox}

\usepackage{tikz}
\usetikzlibrary{calc}
\usetikzlibrary{decorations.pathreplacing}
\usepackage{caption}
\usepackage{subcaption}
\usepackage{booktabs}
\usepackage{cleveref}
\usepackage{graphicx}

\usepackage{multirow}
\newcolumntype{P}[1]{>{\centering\arraybackslash}p{#1}}
\newcolumntype{M}[1]{>{\centering\arraybackslash}m{#1}}

\makeatletter
\renewcommand*\env@matrix[1][\arraystretch]{%
	\edef\arraystretch{#1}%
	\hskip -\arraycolsep
	\let\@ifnextchar\new@ifnextchar
	\array{*\c@MaxMatrixCols c}}
\makeatother

\usepackage{xr-hyper}

\makeatletter
\newcommand*{\addFileDependency}[1]{%
	\typeout{(#1)}%
	\@addtofilelist{#1} %
	\IfFileExists{#1}{}{\typeout{No file #1.}}%
}
\makeatother

\newcommand*{\myexternaldocument}[1]{
	\externaldocument{#1}%
	\addFileDependency{#1.tex}%
	\addFileDependency{#1.aux}%
}

\myexternaldocument{Appendix_24Apr2026} 

\begin{document}
\title{{\textcolor{MyBlue}{\textbf{Individual Shrinkage for Random Effects}}}}
\title{{\textcolor{MyBlue}{\textbf{Individual Shrinkage for Random Effects}}}
	\thanks{We are grateful to Roger Koenker, Laura Liu, Mikkel Plagborg-M$\o$ller, Kirill Ponomarev, Frank Schorfheide, and Suyong Song for valuable comments and suggestions. We benefited from useful comments by participants at several seminars and conferences.
		Silvia Sarpietro acknowledges financial support from the Italian Ministry of University and Research (PRIN Grant 2020B2AKFW).}}

\author{
	Raffaella \textsc{Giacomini} \thanks{%
		University College London. \texttt{r.giacomini@ucl.ac.uk} }\and %
	Sokbae \textsc{Lee} \thanks{%
		Columbia University. \texttt{sl3841@columbia.edu} } \and
	Silvia \textsc{Sarpietro} \thanks{%
		University of Bologna. \texttt{silvia.sarpietro@unibo.it} %
}}

\date{\today}
\maketitle

\begin{abstract}
	This paper develops an approach to random effects estimation and individual-level forecasting in micropanels that targets individual accuracy rather than aggregate performance. The conventional shrinkage methods used in the literature, such as the James-Stein estimator and Empirical Bayes, target aggregate performance and can lead to inaccurate decisions at the individual level. We propose a class of shrinkage estimators with individual weights (IW) that leverage an individual's own history, instead of the cross-sectional dimension. This approach can help overcome the ``tyranny of the majority" inherent in existing methods, while relying on weaker assumptions. A key contribution is addressing the challenge of obtaining feasible weights from short time-series data under parameter heterogeneity. We discuss the theoretical optimality of IW and recommend using feasible weights determined through a Minimax Regret analysis in practice.


\end{abstract}
\clearpage


\section{Introduction}
{\it ``Knowing when to borrow and when not to borrow is one of the key aspects of statistical 
practice" (\cite{mallows1982overview})}

\smallskip

Estimating fixed or random effects (RE) and forecasting individual outcomes are core problems 
in econometrics. In micropanels, however, the short time dimension makes unit-level estimates 
imprecise.\footnote{Estimating individual effects and forecasting with micropanels 
are the goals of several literatures. Examples include ``value added 
models'' for teachers \citep[]{Kane:Staiger:08,CFR:2014:a,CFR:2014:b, Angrist17}, neighborhoods \citep{CH18}, and physicians \citep{FHB14}; 
firm-specific effects in hiring discrimination \citep{kline2022systemic}, and forecasting 
individual incomes for consumption/savings decisions \citep{chamberlain1999predictive}. Macroeconomic panel forecasting also 
falls in this class if it uses short estimation windows to account for parameter 
instability (e.g., \cite{Liu:2020} for forecasting banks' revenues after a regulatory change).} 
This has motivated shrinkage methods, such as 
the James--Stein estimator \citep{JS61} (JS), its extension by \cite{kwon21}, 
and modern Empirical Bayes (EB), which ``borrow strength'' across individuals to 
enhance accuracy. While effective in improving aggregate performance, these methods may lead 
to inaccurate individual decisions, particularly under heterogeneity. 
We propose a complementary shrinkage approach designed to address this limitation and overcome 
the ``tyranny of the majority'' inherent in existing methods.

The distinction between aggregate and individual accuracy is central in many applications. When RE are used to guide policy interventions targeting specific 
individuals---such as teacher dismissal or hospital reward as in the value-added 
literature discussed by \cite{Hull20}, financial distress prediction \citep{Liu:2020}, or 
personalized financial advice \citep{chamberlain1999predictive}---the relevant objective is 
individual accuracy rather than aggregate performance.

Existing shrinkage methods implicitly target aggregate loss. JS, for example, applies the same degree of shrinkage 
to all units, which can induce large bias for outliers; \cite{efron1971limiting} and 
\cite{mallows1982overview} highlight such a ``relevance'' problem, i.e., the assumption that all other individuals are equally relevant for borrowing 
strength.\footnote{\cite{efron1971limiting} suggests tackling the problem by first 
identifying and then not shrinking outliers; \cite{efron2010large} proposes to use covariates 
to identify relevant individuals.} EB relies on stronger assumptions than JS, notably 
exchangeability and a distributional assumption for the errors. When the exchangeability 
assumption is violated, for example if a common RE distribution does not exist, estimates may 
depend on sample composition rather than reflecting true individual quality. Modern EB work
(e.g., \cite{efron2010large}, \cite{gu2015unobserved,gu2023invidious}, \cite{Liu:2020}, 
\cite{chen2022empirical} and \cite{koenker2024empirical}) has substantially relaxed assumptions on the prior; but the impact of misspecification of 
the error distribution remains largely unexplored. We show that such misspecification can 
result in large bias, particularly for outliers.

To address these limitations, we propose a shrinkage approach that explicitly targets 
individual loss and relies on weak assumptions. We introduce a class of shrinkage estimators 
that, like JS, shrinks the time-series 
estimators of RE towards a common mean. Unlike existing approaches, however, our shrinkage 
with individual weights (IW) leverages solely the individual's own history when computing the 
weights, rather than the cross-sectional dimension.

In line with the literature, we consider a model in which individual outcomes are the sum of 
RE and idiosyncratic errors, assumed independent of each 
other.\footnote{\cite{chen2022empirical} has explored relaxing a related 
precision-independence assumption in heteroskedastic EB by allowing the latent 
parameter to depend on known variances, at the cost of imposing additional 
structure.} We are fully agnostic about parameter heterogeneity beyond the common 
mean, which we show is crucial for individual accuracy. In contrast, 
JS assumes homogeneous variances for RE and errors; \cite{kwon21} allows heterogeneous error variances, but maintains a common RE distribution. 
EB methods can allow some forms of heteroskedasticity (e.g., 
\cite{gu2015unobserved}, \cite{Liu:2020}), but typically impose a common RE 
distribution. Our setup encompasses many empirically relevant settings, 
including residual-based formulations arising from panel or value-added models.

An advantage of IW is that it does not rely on a large cross-section for 
accuracy, making it applicable in small samples. However, the use of short panels 
limits the applicability of asymptotic approximations. We therefore focus on finite-sample 
performance and robustness, evaluating estimators based on their performance over the 
parameter space. To this end, we adopt a Minimax Regret criterion, following 
\cite{manski2019econometrics}, which provides a practical and robust decision-theoretic 
framework under limited information. In our setting, this criterion both motivates IW and 
guides the selection of feasible weights.\footnote{Minimax Regret properties of 
shrinkage estimators have been studied by \cite{magnus2002estimation} and 
\cite{hansen2015shrinkage}. Their setting differs from ours in that they study
estimation of the mean of normal variables with known variance. Other work applies Minimax Regret to panel data in 
different contexts, e.g., missing data \citep{Dominitz:Manski:2021} or forecasting 
discrete outcomes under partial identification \citep{CMS:20}.} Since in our baseline model the estimator of RE coincides with the forecast of the individual outcome, we focus on forecasting without loss of insight, and report the analogous estimation results in the Appendix. 

Section~\ref{sec: MR} studies a split-sample setting in which weights and forecasts use non-overlapping information, providing a transparent setting to motivate individual shrinkage. Under the maintained assumptions, we show that IW is Minimax-Regret optimal, relative to the time series forecast and the common mean, and MSFE-optimal when these two benchmarks are equally accurate. Further improvements can arise under an assumption that requires the IW weights to be genuine functions of the RE, which, for example, is not satisfied by JS because its weights are based on cross-sectional information. Finally, we show that IW gains under this assumption increase with the tail heaviness of the RE distribution, as the weights relate shrinkage to how ``far'' the RE is from the common mean, helping mitigate the ``tyranny of the majority.''

The paper's main practical contribution is to develop feasible IW weights for short panels. While it is possible to construct weights that satisfy the sample-splitting assumptions underlying the theory, they are less attractive in short panels because they discard information. We thus focus on the relevant case in which both forecasts and weights use the full sample. We present sample-splitting implementations, consistent with the setting of Section~\ref{sec: MR}, in the Appendix.

We present three feasible weights for IW: estimated oracle weights (IW-O); ``Minimax Regret optimal weights'' (IW-MR); and weights based on the inverse squared forecast error (IW-MSFE), the last analogous to forecast combination weights (e.g., \cite{bates1969combination}, 
\cite{stock1998comparison}), but computed here with short $T$. These weights offer 
additional robustness benefits because they do not rely on correct specification of the model 
and can thus be applied in more general settings. We compare their finite-sample performance 
and find that IW-MR weights perform best, closely followed by inverse MSFE weights. 
Additional simulations show how IW can mitigate the ``tyranny of the majority'' that can affect JS.

There is a substantial literature on forecast combination and pretesting, but most of it assumes large or moderate $T$. For example, \cite{brownlees2025unit} study unit-averaging methods for moderate or large $T$, and 
\cite{pesaran2022forecasting} study optimal forecast combination when $T \to \infty$. In 
contrast, we explicitely focus on short panels and allow heterogeneous variances in both sources 
of unobserved heterogeneity. Closer in spirit, \cite{de2020empirical} also consider minimax 
regret for combining two estimators, but their results are asymptotic and rely 
on one estimator having a smaller asymptotic variance. We instead compare two unbiased 
estimators with different variances without restricting their relative magnitudes, 
and our short-$T$ setting precludes reliance on consistency.

We present two empirical illustrations. The first revisits 
\cite{kline2022systemic} and studies gender discrimination in firm hiring. IW-MR delivers different estimates and policy implications, relative to EB, and performs well in forecasting accuracy and robustness. The second forecasts earnings residuals with the Panel Study of Income Dynamics and shows that 
individual shrinkage can be useful despite large heterogeneity, even in 
terms of aggregate performance.

The rest of the paper is organized as follows. Section~\ref{Sec:loss} introduces the 
framework and discusses the limitations of existing shrinkage methods under individual 
loss. Section~\ref{sec: MR} studies IW in a simplified 
split-sample setting. Section~\ref{sec: MRFeasWeights} derives feasible weights. 
Section~\ref{sec: MC} reports simulation evidence, Section~\ref{sec: Application} presents 
the empirical applications, and Section~\ref{sec: Conclusions} concludes. 
Appendix~\ref{sec: Proofs} contains the proofs, Appendix~\ref{sec: MSE} the estimation results, and Appendix~\ref{sec: MRFeasWeightssimplelagged} the feasible weights for the split-sample setting of Section~\ref{sec: MR}.

\section{Framework and Limitations of Existing Methods}\label{Sec:loss}

This section introduces the framework and shows how shrinkage
methods targeting aggregate loss may deliver inaccurate
individual-level decisions. We discuss the ``tyranny of the majority" and the effects of violating exchangeability or misspecifying
the error distribution.

Suppose that the individual outcomes are the sum of independent RE and errors:
\begin{align}\label{model0}
    Y_{i,t} = A_i + U_{i,t}, \;\;\; i=1,\ldots,N; \;\;\; t=1,\ldots,T,
\end{align}
where $A_i \sim (0, \lambda_i^2)$ and $U_{i,t} \sim (0, \sigma_i^2)$, with a
common zero mean and heterogeneous variances. Let $\bar{Y}_i=\frac{1}{T}\sum_{t=1}^T Y_{it}$ denote the $i$th
time-series mean, i.e., the maximum likelihood estimator (MLE).
We wish to estimate $A = \left(A_1, \ldots, A_N\right)$ or equivalently
forecast $Y_{i, T+1}$ for each $i$ using information up to $T$, via a decision rule $\delta = \left(\delta_1, \ldots,
\delta_N\right)$.\footnote{For notational simplicity we suppress the dependence
on the sample.} Here we focus on estimation, but similar considerations apply to forecasting. Two loss functions are natural:
\begin{align*}
    \text{Individual loss:} \;\;\;\; L_i(\delta_i, A_i)  :=
        \left(\delta_i-A_i\right)^2, \;\;\;\;\;\;
    \text{Aggregate loss:} \;\;\;\; L(\delta, A)  :=
        \frac{1}{N} \sum_{i=1}^N \left(\delta_i-A_i\right)^2.
\end{align*}
If one targets individual loss, restricting attention to linear rules delivers
the optimal rule:\footnote{This can be obtained by applying the ``best linear
rule'' in equation (9.4), page 129 of \cite{efron1973stein}, to
$\bar{Y}_{i}|A_i\sim(A_i,\sigma_i^2/T)$ and
$A_i\sim(0,\lambda_i^2)$.}\label{Efronfootnote}
\begin{align}
\text{Oracle IW}: \delta_i^* =\frac{\lambda_i^2}{\lambda_i^2 +
\sigma_i^2/T}\bar{Y}_i, \label{oraclew}
\end{align}
which shrinks the MLE towards the common zero mean, using
individual-specific weights.

Targeting aggregate loss corresponds to using the posterior mean as the optimal rule. The main Bayesian estimators in this context are the JS estimator of \cite{JS61} and EB estimators such as \cite{Liu:2020}, \cite{efron2016empirical}, and \cite{gu2015unobserved}.

Under model (\ref{model0}), JS does not require distributional
assumptions but assumes homoskedasticity, $\sigma_i^2=\sigma^2$, $\lambda_i^2=
\lambda^2$. \cite{JS61} and \cite{efron1973stein} show that: 
\begin{align}\label{jsrule}
\ \text{JS}: \delta_i^*=\frac{\lambda^2}{\lambda^2+\sigma^2/T} \bar{Y}_i.
\end{align}
has lower aggregate risk, $\mathbb{E}_{A}L(\delta, A)$, than the MLE.
The JS weight is estimated from the cross-section, as $\hat{\lambda}^2/(\hat{\lambda}^2+\hat{\sigma}^2/T)$, where
$\hat{\sigma}^2/T=1/N \sum_{i=1}^N \left[1/(T-1)\sum_{t=2}^T
(Y_{i,t}-Y_{i,t-1})^2\right]/(2T)$ and $\hat{\lambda}^2=1/N
\sum_{i=1}^N(\bar{Y}_i - 1/N \sum_{i=1}^N\bar{Y}_i)^2-\hat{\sigma}^2/T$.
Unlike the oracle IW rule, JS imposes constant
weights across $i$, resulting in the same amount of shrinkage for all
individuals.

EB methods require additional distributional assumptions on $U_{i,t}$
and estimate $A_i$ by the posterior mean. Two main approaches are: ``G-modeling'' and ``f-modeling''. In G-modeling, one estimates the distribution of $A_i$, $G(A)$, from the cross-section, and then computes the posterior mean. \cite{efron2016empirical} assumes normal homoskedastic errors and estimates  $G(A)$ via deconvolution; \citet{gu2015unobserved} use a nonparametric MLE based on the Kiefer-Wolfowitz approach for mixture models and assume normal errors that can be heteroskedastic
(treating the variances as random). Under normality and
homoskedasticity the G-modeling-rule is:
\begin{align}\label{g_mod}
    \text{EB (G-modeling)}: \delta^*_i= \frac{\int A \varphi _{i}\left(
    \bar{Y}_{i}-A \right) dG\left( A \right) }{\int \varphi _{i}\left(
    \bar{Y}_{i}-A \right) dG\left( A \right) },
\end{align}
with $\varphi _{i}\left( \cdot \right)$ the probability density function of a
$\mathcal{N}\left(0,\sigma^{2}/T\right)$. F-modeling bypasses estimation of $G(A)$ and instead estimates the posterior mean through Tweedie's correction (available
for exponential family errors in \cite{efron2011tweedie}). The correction depends on the marginal density, which can be estimated
nonparametrically from the cross-section. Let $l'(\bar{Y}_i)=\frac{\partial}{\partial \bar{Y}_i} \log f(\bar{Y}_i)$
and $f(\bar{Y}_i)$ the marginal density of $\bar{Y}_i$, under normal and homoskedastic errors, the rule is (see e.g. \citet{Liu:2020}, with an extension to some forms of heteroskedasticity in their Section 5):
\begin{align}\label{f_mod}
    \text{EB (f-modeling)}: \delta^*_i= \bar{Y}_i+ \frac{\sigma^2}{T}
    l'(\bar{Y}_i).
\end{align}

\medskip
\noindent
\textbf{The Tyranny of the Majority:}
\cite{efron1971limiting}, \cite{mallows1982overview} and \cite{efron2010large}
discuss the notion of ``relevance'': which units are informative for estimating
a given $A_i$. 
\cite{efron2010large} shows that ignoring relevance can lead to bias.
To see why JS can yield inaccurate individual-level decisions, assume $\sigma^2_i=\sigma^2$. 
The individual-level bias from using JS instead of Oracle IW is:
$\mathrm{Bias}_i = \left(\frac{\lambda_i^2}{\lambda_i^2+\sigma^2/T} -
\frac{\lambda^2}{\lambda^2+\sigma^2/T}\right)\bar{Y}_i.$
When $\lambda_i^2 \neq \lambda^2$, the bias is large for large $\bar{Y}_i$, e.g., outliers, illustrating the
``tyranny of the majority". EB may be less susceptible to this issue under correct specification of the error distribution (see Figure 1 in \cite{Liu:2020}).

\medskip
\noindent
\textbf{Violation of Exchangeability:}
A related issue is exchangeability, under which the joint distribution of the data is invariant to permutations of
the indices. Standard EB methods impose a common RE distribution and i.i.d.\ observations, sufficient for exchangeability. In our fully heterogeneous model (\ref{model0}), with normal
errors: $U_{i,t} \sim \mathcal{N}(0, \sigma_i^2)$, the marginal likelihood is
\begin{align}\label{eq:our_likelihood}
    f_{i}( Y_{i,1},\ldots,Y_{i,T} )
    &:= \int \ell (Y_{i,1},\ldots,Y_{i,T} | \lambda_i^2, \sigma_i^2) dH_i (
        \lambda_i^2, \sigma_i^2) \nonumber \\
    &= \int \int \left( \sqrt{2 \pi \sigma_i^2} \right)^{-T} \exp \left[ -
        \frac{1}{2\sigma_i^2} \sum_{t=1}^T (Y_{it} - A_i)^2 \right]
    dL_i (A_i | \lambda_i^2 ) dH_i ( \lambda_i^2, \sigma_i^2),
\end{align}
where $L_i(\cdot | \lambda_i^2 )$ is the $i$-specific distribution of $A_i$ and
$H_i(\cdot, \cdot)$ that of $(\lambda_i^2,
\sigma_i^2)$. This likelihood is too heterogeneous to be embedded in the EB approach, which thus requires more assumptions. First, one typically
assumes identical $L_i$ and $H_i$ across $i$. Then, the marginal
likelihood becomes
\begin{align}\label{eq:NPEB_likelihood}
    f_{G}( Y_{i,1},\ldots,Y_{i,T} )
    &= \int \int \left( \sqrt{2 \pi \sigma_i^2} \right)^{-T} \exp \left[ -
        \frac{1}{2\sigma_i^2} \sum_{t=1}^T (Y_{it} - A_i)^2 \right]
    dG ( A_i, \sigma_i^2),
\end{align}
where $G( a, \sigma^2 ) := L (a) H (\sigma^2)$ is the unknown but common
distribution. Second, EB methods typically assume that the random components $(
A_i, \sigma_i^2)$ are i.i.d.\ over $i$, so that $G$ can be estimated
nonparametrically. 
Differences between (\ref{eq:our_likelihood})
and (\ref{eq:NPEB_likelihood}) imply that EB might lead to incorrect decisions in our setting with heterogeneous $\lambda_i^2$ and no exchangeability
assumptions.

\medskip
\noindent
\textbf{Misspecification of the Error Distribution:}
EB methods rely on a complete model, including a specific error distribution.
We illustrate the, largely unexplored, impact of misspecifying this distribution on EB
estimators based on the Tweedie correction, through a simple example.
Assume $T=1$ and drop the $t$ subscript. Let $U_{i}$
follow a standardized Gamma distribution with mean zero, variance one, and skewness $\gamma$.
We derive the misspecification bias of the EB estimator if wrongly
assuming standard normal errors:\footnote{Tweedie's formula under standard normal errors gives the posterior mean as:
$E[A_i|Y_i] = Y_i+ l'(Y_i),$ where $l'(Y_i)=\frac{\partial}{\partial
\bar{Y}_i} \log f(Y_i)$ and $f(Y_i)$ is the marginal distribution of $Y_i$.
Using \cite{efron2011tweedie}, 
a Gamma distribution with
shape parameter $m$ (and skewness $\gamma \equiv 2/\sqrt{m}$), zero mean
and unit variance implies a posterior mean:
$E[A_i|Y_i] = \frac{Y_i+\gamma/2}{1+\gamma Y_i/2}+ l'(Y_i).$
The bias is obtained by subtracting the two
expressions for $E[A_i|Y_i]$.}
$
\mathrm{Bias}_i = \frac{1- Y_i^2}{2/\gamma + Y_i}.$
Thus, when $|Y_i|$ is large, e.g., for outliers, the bias is large in absolute value (approximately $\gamma/2$ when $Y_i$ is near the zero mean), another manifestation of the ``tyranny of the majority", here due to misspecification.

\section{Motivation for Shrinkage with Individual Weights (IW)}\label{sec: MR}

To address the issues described in Section~\ref{Sec:loss}, we focus on
individual loss and construct shrinkage weights from each individual's history rather than from the cross-section. This mitigates the relevance problem and the ``tyranny of the majority.''
By relaxing exchangeability and parametric assumptions on the error term, our approach is less sensitive to potential misspecification.

In particular, for each individual $i$, we consider the shrinkage rule $\widehat{Y}^{IW}_{i,T}$ 
with weights $W_{i,T}$ that are based only on individual time-series observations up to time $T$:
\begin{align}
\text{Shrinkage with Individual Weights (IW)}: \widehat{Y}^{IW}_{i,T} &=
    \widehat{Y}^{TS}_{i,T}W_{i,T}+\mu(1-W_{i,T}).\label{generalIW}
\end{align}
The point of shrinkage $\mu$ is either known (e.g., $\mu=0$ if the data are demeaned or residuals from a first-step estimation
that includes an intercept) or approximated by the pooled mean,
$\mu= \sum_{i=1}^N\sum_{t=1}^T Y_{i,t}/NT$. $\widehat{Y}^{TS}_{i,T}$ is an
estimator of the RE or a forecast made at time $T$ for the outcome $Y_{i,T+1}$
that is only based on the time series dimension.

In this section, we motivate IW by studying its theoretical properties in a simplified split-sample setting, where weights and forecasts use non-overlapping information. This yields a transparent finite-sample analysis.  We focus on forecasting $Y_{i,T+1}$ and characterize when IW is optimal at the individual level in terms of MSFE and Minimax Regret. Appendix~\ref{sec: MSE} gives the analogous results for estimation of the RE.

\subsection{Model and Assumptions}
The model is:
\begin{align}\label{model}
    Y_{i,t} = A_i + U_{i,t}, \;\;\; i=1,\ldots,N; \;\;\; t=1,\ldots,T+1,
\end{align}
where $A_i \sim (\mu, \lambda_i^2)$ and $U_{i,t} \sim (0, \sigma_i^2)$. Here,
$A_i, U_{i,1},\ldots,U_{i,T+1}$ are random variables, whereas $\mu$,
$\lambda_i^2$ and $\sigma_i^2$ are parameters. In other words, we take a
frequentist approach. We wish to forecast $Y_{i, T+1}$ for each individual $i$
using information up to time $T$.

The simplified split-sample setting assumes that the time series forecast is the
time-$T$ outcome and the IW weights are based on data up to time $T-1$ rather
than $T$, so that:
\begin{align}
\text{IW}: \widehat{Y}^{IW}_{i,T}&=\widehat{Y}^{TS}_{i,T} {W}_{i,T-1} +
    \widehat{Y}^{Pool}_{i,T} (1-{W}_{i,T-1}), \label{IWsimple}\\
\widehat{Y}^{TS}_{i,T}&=Y_{i,T}, \nonumber \\
\widehat{Y}^{Pool}_{i,T}&= \mu. \nonumber
\end{align}
We make the following assumptions.
\begin{assumption}\label{indep-assumption}
    $A_i, U_{i,1},\ldots,U_{i,T+1}$ are mutually independent.
\end{assumption}
\begin{assumption}\label{key-assumption}
The individual weight ${W}_{i,T-1}$ satisfies $0 \leq {W}_{i,T-1} \leq 1$ and
\begin{align}\label{key:regularity:IW}
\mathrm{Cov} \left\{  (A_i-\mu)^2,  \left(1 - {W}_{i,T-1}\right)^2
\right\} \leq  0.
\end{align}
\end{assumption}

\begin{remark}[Assumptions]
Assumption~\ref{indep-assumption} is standard in the EB literature and could be relaxed under additional structure. For instance, \cite{chen2022empirical} relaxes a related assumption of
independence between the latent parameter and known variances, by imposing
additional structure on their conditional relationship.
Assumption~\ref{indep-assumption} is plausible in a range of applications,
including settings with systematic firm-level heterogeneity, as in our empirical
application. 

Assumption~\ref{key-assumption} states that $(A_i-\mu)^2$ and $\left(1 - {W}_{i,T-1}\right)^2$ are weakly
negatively correlated, so, for instance, larger values of
$(A_i-\mu)^2$ receive smaller weights to the pooled
forecast (or are uncorrelated). It rules out ``pathological'' weights that would shrink outliers more than
units at the center of the distribution, thus exacerbating the ``tyranny of the majority" that we seek to overcome. Constant weights, i.e., ${W}_{i,T-1} = c_i$ for some $0 \leq c_i
\leq 1$, satisfy Assumption~\ref{key-assumption} with equality. If the individual weight is a genuine
function of the RE, the inequality in
\eqref{key:regularity:IW} can be strict. We will show that this strict
inequality translates into performance gains for IW.
\end{remark}

\begin{remark}[Interpretation of $\mu$] \label{muint}
The common mean $\mu$ is the point of shrinkage and represents how we borrow
strength from the majority. As in a
classical Bayesian setting, we treat it as a tuning parameter and, thus, similarly to existing approaches, our theoretical results abstract from uncertainty in its estimation.
As discussed by \cite{kwon21}, in empirical work outcomes are often demeaned, so $\mu=0$ (e.g., if $Y_{i,t}$ are
residuals from a first-stage estimation of a model with an intercept, see
Remark~\ref{cova}). If $\mu$ is unknown, we replace it with the panel
mean of $Y_{i,t}$. Remark~\ref{groupstructure} discusses how $\mu$ could be chosen in the case of a known group structure in parameters. 
\end{remark}

\begin{remark}[Extensions: covariates and value-added models] \label{cova}
	Covariates can be incorporated by redefining $Y_{i,t}$ in (\ref{model}) as residuals from a first-step estimation with homogeneous coefficients:
	\begin{equation}\label{panelres}
		Y_{i,t}=\tilde{Y}_{i,t}-X_{i,t}'\widehat{\beta},
	\end{equation}
	where $\tilde{Y}_{i,t}$ are the outcomes and $\widehat{\beta}$ is consistent as $N\to\infty$.\footnote{If $X_{i,t}$ includes lagged outcomes, one could use the Arellano--Bond estimator \citep{arellano1991some}.} The theoretical results below then apply under the additional assumption that $N$ is large. In finite $N$, the consistency requirement could be relaxed if alternative, possibly biased, estimators improve forecast accuracy.
	
	Value-added models can be handled similarly. If $
	\tilde{Y}_{i,j,t}=X_{i,j,t}'\beta+A_i+U_{i,j,t},$
	where, for example, $i$ indexes teachers and $j=1,\ldots,n_{i,t}$ indexes students assigned to teacher $i$ at time $t$, then the model is nested in (\ref{model}) by defining
	$	Y_{i,t}=\frac{1}{n_{i,t}}\sum_{j=1}^{n_{i,t}}\tilde{Y}_{i,j,t}
		-\frac{1}{n_{i,t}}\sum_{j=1}^{n_{i,t}}X_{i,j,t}'\widehat{\beta},$ provided $\widehat{\beta}$ is consistent as $N\to\infty$.
 Extending the analysis to heterogeneous slopes would turn the univariate problem studied here into a multivariate one, which we leave for future work.
\end{remark}

\begin{remark}[Robustness to distributional assumptions]
We make no distributional assumptions on RE and idiosyncratic errors.
Heavy tails in both distributions are permitted, as long as the variances exist
(i.e., the parameters $\lambda_i^2$ and $\sigma_i^2$ are finite).
\end{remark}

\begin{remark}[Robustness to dependence structure]
Because the analysis is individual-level, it does not require large $N$ or restrictions on cross-sectional
dependence when $\mu$
is known. 
If $\mu$ is instead approximated by the sample mean, restrictions on cross-sectional dependence are needed for a law of large numbers. Incorporating covariates as in Remark~\ref{cova} similarly requires conditions ensuring a consistent estimator of the homogeneous coefficients.
Time-series dependence can be handled through lagged dependent
variables with homogeneous autoregressive coefficients. While this specification may not capture
other forms of temporal dependence, such as moving-average or volatility
dynamics, these likely play a limited role in short time series.
\end{remark}

\begin{remark}[Robustness to distribution of parameters across $i$]
We are purposely agnostic
about the distribution of $\lambda_i^2$ and $\sigma_i^2$ across $i$, so, in general, we cannot make formal statements about aggregate
performance. Nonetheless, Section~\ref{grouperf} discusses the implications of our findings for aggregate accuracy. Also, while we assume
independence between RE and errors, we accommodate arbitrary dependence between their variances (e.g., there could be two groups of
units, one with low $\lambda_i^2$ and low (high) $\sigma_i^2$ and one
with high $\lambda_i^2$ and high (low) $\sigma_i^2$). 
\end{remark}

\begin{remark}[Known group structure in parameters] \label{groupstructure}
Suppose there is a group structure in $\mu$, with a finite number of subgroups
and observable group membership (with $\lambda^2_i$ and $\sigma^2_i$ still
heterogeneous within the subgroups). Then the only change to our analysis is that the point of shrinkage becomes the subgroup mean. If homogeneity within subgroups extends to
$\lambda^2_i$ and $\sigma^2_i$, IW reduces to JS applied to each subgroup (and it is exactly JS with one group only).
\end{remark}

Henceforth, we focus on model (\ref{model}), with the understanding that $Y_{i,t}$ are either raw outcomes or residuals such as (\ref{panelres}) (in a large-$N$ setting).

\subsection{MSFE and Minimax Regret}\label{sec:basic}
This section discusses the two criteria that we use to evaluate the performance
of IW: MSFE and Minimax Regret. Consider a situation where there is uncertainty
about the parameter $\theta_i = (\lambda^2_i,\sigma^2_i)$. The MSFE of forecast
$m \in \mathcal{M}$ for a given $\theta_i$ is
$\mathrm{MSFE}(m, \theta_i)=\mathbb{E} \left[ \left( Y_{i,T+1} -
\widehat{Y}^{m}_{i,T} \right)^2 \right].$

The next lemma derives the MSFEs of TS, Pool, and IW in~(\ref{IWsimple}).
\begin{lemma}\label{thm:0}
Consider the forecasts in~(\ref{IWsimple}). Then under
Assumption~\ref{indep-assumption} we have
\begin{align*}
 \mathrm{MSFE}(\mathrm{TS}, \theta_i) &= 2 \sigma_i^2, \\
 \mathrm{MSFE}(\mathrm{Pool}, \theta_i) &=  \lambda_i^2 + \sigma_i^2, \\
\mathrm{MSFE}(\mathrm{IW}, \theta_i)
&=
\sigma_i^2 + \sigma_i^2 \mathbb{E} \left[ {W}_{i,T-1}^2   \right]
+ \mathbb{E} \left[ (A_i - \mu)^2
\left(1 - {W}_{i,T-1}\right)^2   \right].
\end{align*}
\end{lemma}
\vspace{-.3cm}
Lemma~\ref{thm:0} suggests that the trade-off between TS and Pool in terms of
MSFE depends on the ``signal-to-noise'' ratio $\lambda_i^2/\sigma_i^2$: Pool
dominates when the ratio is less than 1 and TS dominates when it is above 1. Since the parameters are unknown, it is not possible to
choose a forecast optimally. We thus pursue an alternative route. We seek a
robust rule that performs well over the entire parameter space, in the sense of
avoiding large errors when TS and Pool have different accuracy and improving on
both TS and Pool when they have similar accuracy. The following sections show
that IW can accomplish both goals.

We first formalize the notion of robustness that we consider here, based on the
Minimax Regret criterion. Let $\mathcal{M}$ include $\mathrm{TS}$,
$\mathrm{Pool}$, and $\mathrm{IW}$. We define regret as
\begin{align}\label{def:regret:1}
R(m,\theta_i) := \mathrm{MSFE}(m, \theta_i) - \min_{h \in \mathcal{M}}
\mathrm{MSFE}(h, \theta_i).
\end{align}
The Minimax Regret (MMR) criterion selects the forecast $m$ that minimizes the
maximum regret $
\max_{\theta_i \in \Theta} R(m,\theta_i)$,
where $\Theta$ is compact. This notion is close to regret in decision theory without sample data (e.g., see equation (3)
in \cite{manski2019econometrics}). The MMR criterion is championed by
\cite{manski2019econometrics}.\footnote{See Section A.2 in
\cite{manski2019econometrics} and references therein for a detailed discussion.}
The regret in~\eqref{def:regret:1} is defined relative to the best forecast (in
terms of MSFE) out of a set of three because the goal in this section is to
choose among IW, TS, and Pool.\footnote{Under a Minimax criterion instead of MMR, one would have the trivial solution that TS is preferred to Pool if $\max_{i} \sigma_i^2 < \max_{i} \lambda_i^2$
and vice versa. In that case, IW need not minimize the maximum MSFE.}

We note that, since regret is defined relative to a benchmark class, expanding this class may change the regret-optimal forecast, even if no added forecasts is itself optimal.
Our definition of MMR optimality is close in spirit to that in
\cite{de2020empirical}.

\subsection{Minimax Regret Optimality of IW}

We characterize when IW is MMR-optimal. We restrict our attention to the parameter space shown in
Figure~1, where the signal-to-noise ratio ranges
from $1-\nu$ to $1+\nu$ for some $0 \leq \nu < 1$:
\begin{align}\label{def:state}
\Theta = \Theta (\nu) := \{ (\sigma_i^2, \lambda_i^2) \in \mathbb{R}_{+}^2:
1- \nu \leq \lambda_i^2/\sigma_i^2 \leq 1+\nu \}.
\end{align}

\vspace{-.3cm}
\begin{figure}[t]
\centering
\includegraphics[scale=0.5]{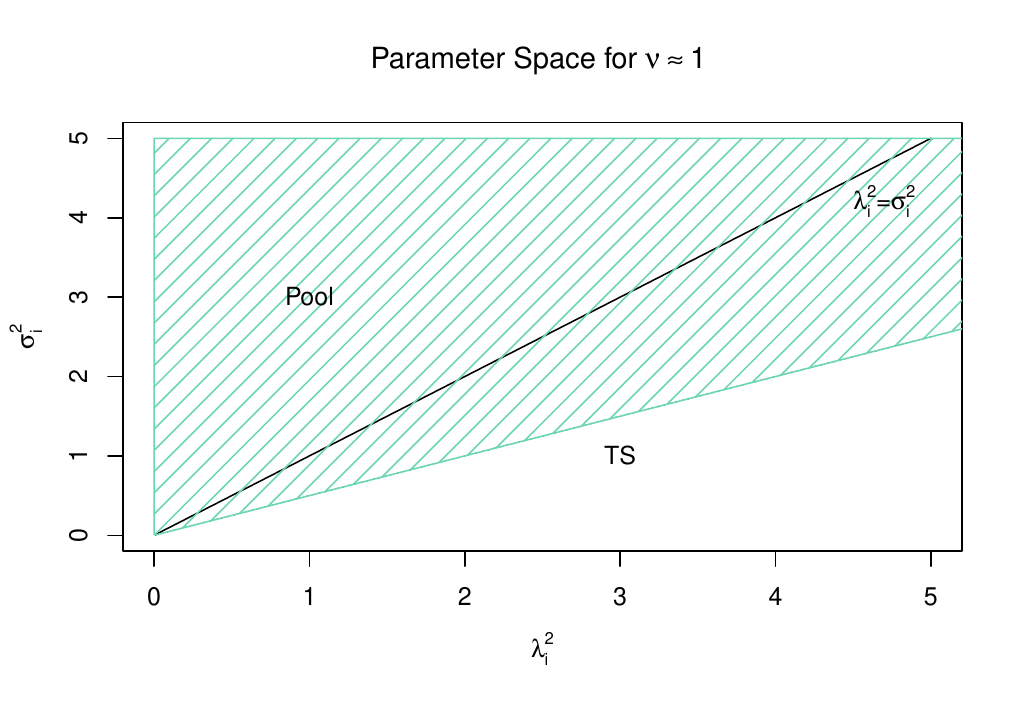}
\caption{Parameter space $\Theta(\nu)$ for $\nu \approx 1$, with $\lambda_i^2$ ($\sigma_i^2$)
on the horizontal (vertical) axis. The shaded region satisfies $1-\nu \le \lambda_i^2/\sigma_i^2 \le
1+\nu$. The diagonal $\lambda_i^2=\sigma_i^2$ corresponds to equal MSFE of TS and Pool; above (below) the diagonal, Pool (TS) has lower MSFE.}
\label{fig:StateSpace}
\end{figure}

Considering a neighbourhood of 1 is natural, since $\lambda_i^2=\sigma_i^2$  represents the
case where TS and Pool are equally accurate. The radius of the neighbourhood is
constrained by the fact that the signal-to-noise ratio cannot be negative, so in
practice we only exclude cases where TS strongly dominates, due to large variance
of the RE and low variance of the error.\footnote{Using common $\nu$
for both bounds is for convenience only. Figures~\ref{fig:msfe} and~\ref{fig:mmr} show that we
are conservative, as increasing the upper bound on $\lambda_i^2/\sigma_i^2$
would not change the conclusions of the Minimax Regret analysis.}

The next theorem shows that IW (uniquely) minimizes maximum regret among TS,
Pool, and IW under Assumptions~\ref{indep-assumption}
and~\ref{key-assumption}.
\begin{theorem}\label{cor:mmr}
Let Assumptions~\ref{indep-assumption} and~\ref{key-assumption} hold. Then,
\begin{align*}
\max_{\theta_i \in \Theta}  R(\mathrm{IW},\theta_i)
\leq \min \left\{
\max_{\theta_i \in \Theta}  R(\mathrm{TS},\theta_i),
\max_{\theta_i \in \Theta}  R(\mathrm{Pool},\theta_i)
\right\},
\end{align*}
where $\Theta$ is defined in~\eqref{def:state}. Furthermore, the inequality
above is strict if either $0 < {W}_{i,T-1} < 1$ with positive probability or
the inequality in~\eqref{key:regularity:IW} is strict.
\end{theorem}
The improvement of IW over TS and Pool in terms of regret
is strict, for example, with any constant weight strictly in $(0,1)$, and larger, keeping all else equal, if the weight is a genuine
function of the RE and Assumption~\ref{key-assumption} holds with a strict
inequality. JS, for example, delivers
weights that are strictly between 0 and 1 but do not depend on $A_i$. Thus, JS outperforms TS and Pool (in terms of MMR) for
all units, yet can itself be improved upon by any admissible weight that responds to the RE satisfying
Assumption~\ref{key-assumption}. The
theorem therefore illustrates the potential benefits of individual weights that are
based on the time series dimension, and thus capture the RE,
relative to existing shrinkage approaches that leverage the
cross-section.

We illustrate the findings of Theorem~\ref{cor:mmr} in
Figures~\ref{fig:msfe} and~\ref{fig:mmr}. Consider one individual (so drop the
subscript $i$) observed over 4 time periods, with $U_{1},\ldots,U_{4}$ drawn
independently from $\mathcal{N}(0, 1)$ and $A$ drawn from $\mathcal{N}(0,
\lambda^2)$. Repeating the simulation many times approximates the individual MSFE and regret when forecasting $Y_4$ at $T=3$
using TS, Pool, or IW. Figure \ref{fig:msfe_regret} plots these MSFEs and regrets as a
function of the signal-to-noise ratio. For IW, we use the feasible Minimax
Regret optimal rule (IW-MR) derived in equation~(\ref{MR}) below but
specialized to the simplified setting in this
section.\footnote{\label{footIWMR} Specifically, we have
$\widehat{Y}^{TS}_{3}=Y_3$, $\widehat{Y}^{Pool}_{3}=0$ and
$\widehat{Y}^{IW-MR}_{3}=Y_{3} {W}_{2}$, with $W_2= 1 -
1/\sqrt{\frac{ \max \{ Y_{1}^2,Y_{2}^2 \}}{0.5(Y_{1} - Y_{2})^2  } + 1}$.
This is the feasible IW-MR derived in Appendix~\ref{sec:
MRFeasWeightssimplelagged}.}
Figure~\ref{fig:msfe} shows that no rule uniformly dominates in terms of
MSFE, but IW is the most accurate over most of the parameter space, except when the
signal-to-noise ratio is very small. Figure~\ref{fig:mmr} shows that IW
is MMR-optimal over the parameter space, since it has the smallest maximum regret among the three
rules: the maximum regret for TS
(dashed line) is around 1 (when $\lambda^2/\sigma^2$ is close to zero), for Pool (dotted line) is around 1.4 (when $\lambda^2/\sigma^2$ is large), and for IW
(solid line) is 0.27 (when $\lambda^2/\sigma^2$ is close to zero). 

%

\begin{figure}[t]
	\centering
	\begin{subfigure}{0.45\textwidth}
		\centering
		\includegraphics[width=\textwidth]{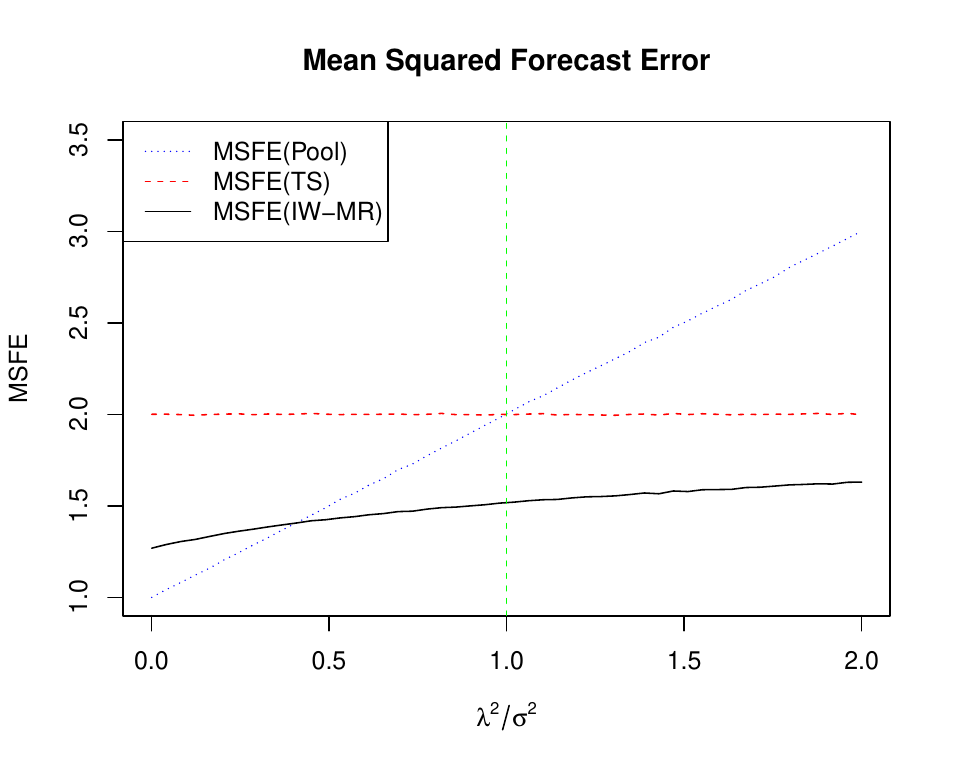}
		\caption{MSFE}
		\label{fig:msfe}
	\end{subfigure}
	\hfill
	\begin{subfigure}{0.45\textwidth}
		\centering
		\includegraphics[width=\textwidth]{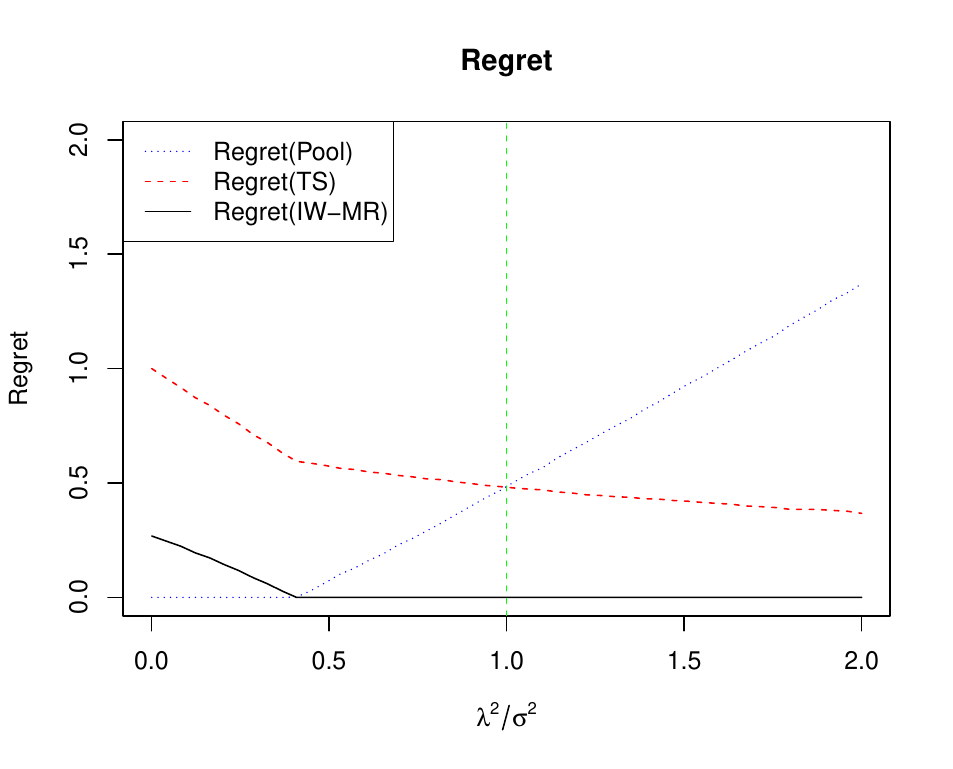}
		\caption{Regret}
		\label{fig:mmr}
	\end{subfigure}
	\caption{MSFE and Regret of TS, Pool, and feasible IW (IW-MR) as functions of the signal-to-noise ratio $\lambda^2/\sigma^2$.}
	\label{fig:msfe_regret}
	
\end{figure}

\subsection{MSFE Optimality of IW}

Figure~\ref{fig:msfe} suggests that IW outperforms TS and Pool in terms of MSFE when the signal-to-noise ratio is near 1. 
The next theorem
shows that IW outperforms TS and Pool in terms of MSFE regardless of
the data-generating process, when the ratio equals 1. Thus, IW is not only
robust, i.e., MMR-optimal, but also optimal in terms of
MSFE when TS and Pool are equally accurate and thus would be indistinguishable.

\begin{theorem}\label{thm:2:IWeq}
Let Assumptions~\ref{indep-assumption} and~\ref{key-assumption} hold. Suppose
that $\lambda_i^2 = \sigma_i^2$. Then,
\begin{align*}
    \mathrm{MSFE}(\mathrm{IW}, \theta_i)
    &\leq \mathrm{MSFE}(\mathrm{TS}, \theta_i)
    = \mathrm{MSFE}(\mathrm{Pool}, \theta_i)
    = 2 \sigma_i^2.
\end{align*}
Furthermore, the inequality above is strict if either $0 < {W}_{i,T-1} < 1$
with positive probability or the inequality in~\eqref{key:regularity:IW} is
strict.
\end{theorem}

The theorem shows that IW is weakly more accurate than TS and Pool when the two
forecasts have equal accuracy. As in Theorem~\ref{cor:mmr}, accuracy gains can
arise when the weights are strictly between 0 and 1 or genuine
functions of the RE. Thus, considering individual weights that leverage the time series to capture the RE can strictly improve accuracy in terms of MSFE.

\subsection{Accuracy Gains and Tail Heaviness}
We now illustrate Theorem~\ref{thm:2:IWeq} and show how the accuracy gains of IW are
linked to the heaviness in the tails of the RE distribution. As in
Figure~\ref{fig:msfe_regret}, consider one individual observed
over 4 periods, but now set $\sigma^2=\lambda^2=1$ so TS and Pool are equally accurate. Let $U_{1},\ldots,U_{4}$ be i.i.d. $\mathcal{N}(0, 1)$ and draw $A$ from a Pareto with different degrees of tail heaviness (chosen so that $\sigma^2=\lambda^2$).\footnote{We consider the double Pareto distribution
with pdf $f(x; \theta; \beta)=\theta/(2\beta) \begin{cases}
    (x/\beta)^{\theta-1}, & \text{if $0<x<\beta$} \\
    (\beta/x)^{1-\theta}, & \text{if $x \geq \beta$}
\end{cases}$ with the following pairs of shape ($\theta$)
and scale ($\beta$) parameters: $(2.3, .5)$, $(3,1)$, $(5, 2.45)$, $(50,
34.5)$. Population moments of order $\geq \theta$ do not exist.
We quantify the tail heaviness of $A_i$ by a quantile-based measure of kurtosis, Crow-Siddiqui
($CS$)$= (Q_{0.975}-Q_{0.025})/(Q_{0.75}-Q_{0.25})$. For each pair, CS is respectively 7.58, 6.59, 5.51, 4.42.\\}
Repeating the simulation many times approximates the
individual MSFE when forecasting $Y_4$ at $T=3$ using TS, Pool, or
IW (IW-MR in
footnote~\ref{footIWMR}). Figure~\ref{fig:mmr_equallyaccurate} plots the MSFEs
of TS, Pool, and IW as a function of tail heaviness in the
distribution of $A$, as measured by Crow-Siddiqui kurtosis ($x$-axis). IW improves on TS and
Pool when equally accurate, confirming Theorem~\ref{thm:2:IWeq}. Furthermore, IW gains are larger the
heavier the tails.

\begin{figure}[t]
    \centering
        \includegraphics[scale=0.5]{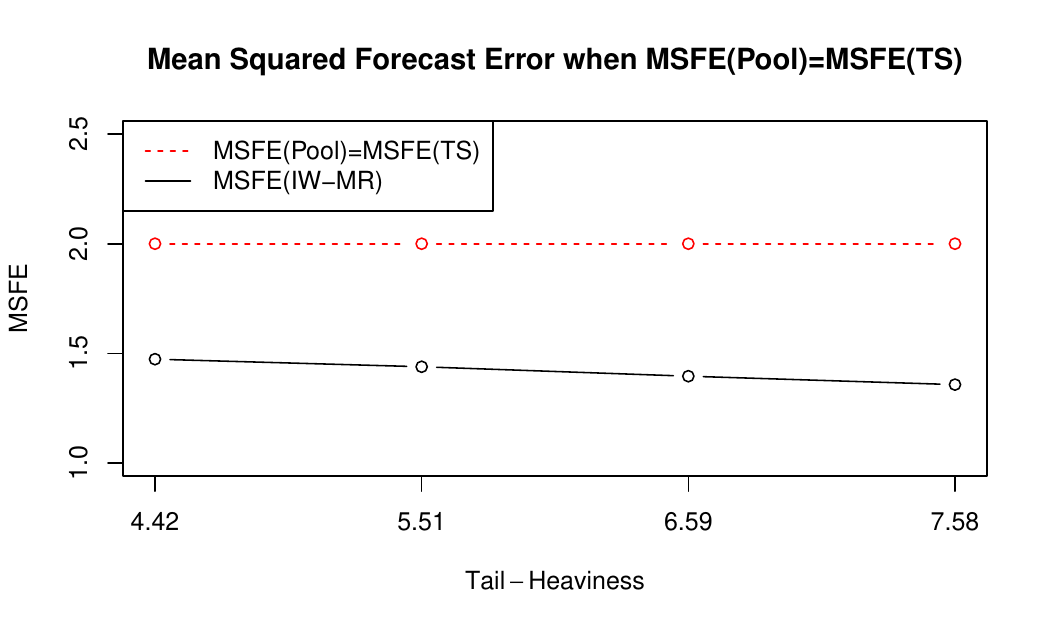}
        \caption{MSFE of feasible IW, Pool, and TS, as a function of RE tail
        heaviness when the MSFE of TS and Pool coincide.}
        \label{fig:mmr_equallyaccurate}
\end{figure}

To illustrate how Assumption~\ref{key-assumption} is linked to tail heaviness,
we consider the same simulation design and compute the covariance in
Assumption~\ref{key-assumption}, focusing on IW-MR only. The
four distributions have increasing tail heaviness, while everything else that
would otherwise affect the weights is kept fixed. We find that this covariance
decreases in the tail heaviness (it respectively equals
$-0.134$, $-0.168$, $-0.216$, $-0.254$ for the four distributions).
Heavy tails make the inequality in
Assumption~\ref{key-assumption} stricter, which, as shown by
Theorems~\ref{cor:mmr} and~\ref{thm:2:IWeq}, translates into larger gains of IW
relative to TS and Pool.

\subsection{Implications for Aggregate Performance}\label{grouperf}
The previous findings show the benefits of IW in terms of
individual performance, but they also have implications for aggregate
performance. Figure~\ref{fig:msfe} provides some intuition for
how the distribution of $\lambda^2_i/\sigma^2_i$
across $i$ (which we leave unrestricted) affects aggregate performance
as measured by the average MSFE. If enough individuals lie in the range where IW dominates individually, IW will also dominate in terms of average MSFE. Our
simulations below will also show how the tail properties of the RE distribution
can be linked to improved performance of IW relative to shrinkage estimators, and that IW can be beneficial not only in the tails but
also for units near the mean of the distribution. Therefore, gains in aggregate performance of IW relative to existing methods
depend on how many individuals lie in these regions.

\section{Feasible Weights for IW}\label{sec: MRFeasWeights}

The results in Section~\ref{sec: MR} are derived in a split-sample setting in
which weights and forecasts are constructed from non-overlapping information; the
corresponding feasible implementations are presented in
Appendix~\ref{sec: MRFeasWeightssimplelagged}. This setting aligns closely with
the assumptions underlying the theory and provides a fully
theory-consistent implementation based on sample splitting.

In this section, we instead consider the full-sample case in which both forecasts
and weights use overlapping information. This case is of primary
practical interest in short panels, where discarding observations through sample
splitting entails a non-negligible loss of information. The use of overlapping
information introduces dependence between weights and forecast errors, which we
characterize theoretically. We then propose feasible implementations that abstract
from this dependence in practice, as the corresponding terms are difficult to
estimate in our short time-series setting. We derive three
classes of feasible individual weights $W_{i,T}$ for equation~(\ref{generalIW}), with $\textrm{TS} =
\bar{Y}_{i,T}=\sum_{t=1}^T Y_{i,t}/T$ and ${W}_{i,T}$ based on time-series
data up to time $T$.

\subsection{Estimated Oracle Weights (IW-O)}\label{FeasibleIW_simple}

The oracle weights minimizing the individual MSFE of IW,
$\mathrm{MSFE}(\widehat{Y}^{IW}_{i,T})=\mathbb{E} \left[ \left( Y_{i,T+1} -
\widehat{Y}^{IW}_{i,T} \right)^2 \right],$ are function of the individual variance parameters:
\begin{align}\label{owgen}
    W^{o}_{i,T} = \frac{\lambda_i^2}{\lambda_i^2 + \sigma_i^2/T}.
\end{align}
They follow from equation (9) in Chapter 4 of \cite{timmermann2006forecast},
since the joint distribution of $Y_{i,T+1}$ and
$\widehat{Y}^{TS}_{i,T}$ is
$$\begin{pmatrix}[0.8]
    Y_{i,T+1}\\
    \widehat{Y}^{TS}_{i,T}
\end{pmatrix} \sim \left( \begin{bmatrix}[0.8]
    \mu\\
    \mu
\end{bmatrix}, \begin{bmatrix}[0.8]
    \lambda_i^2+\sigma_i^2 & \lambda_i^2\\
    \lambda_i^2 & \lambda_i^2+\sigma_i^2/T
\end{bmatrix} \right),$$
which gives the optimal weight on $\widehat{Y}^{TS}_{i,T}$ as inverse of the forecast variance times the covariance between
the outcome and the forecast. The same expression follows from the ``best linear rule'' in equation (9.4) of \cite{efron1973stein} with $\widehat{Y}^{TS}_{i,T}|A_i\sim(A_i,\sigma_i^2/T)$ and
$A_i\sim(\mu,\lambda_i^2)$.

The first set of feasible weights are based on these oracle weights and can be obtained as:
\begin{align}\label{optimal-feasible}
    \frac{ \sum_{t=1}^T (Y_{i,t}-\mu)^2 /T
    - \sum_{t=1}^{T-1} (Y_{i,t} - Y_{i,t+1})^2 /2(T-1)}
    {\sum_{t=1}^T (Y_{i,t}-\mu)^2 /T
    - \sum_{t=1}^{T-1} (Y_{i,t} - Y_{i,t+1})^2 /2T},
\end{align}
using the facts that: $\sum_{t=1}^T (Y_{i,t}-\mu)^2 /T$ is an unbiased
estimator of $\lambda_i^2+\sigma_i^2$; the denominator of the oracle weights can
be rewritten as $\lambda_i^2+\sigma_i^2-\frac{T-1}{T}\sigma_i^2$; and that
$\widehat{\sigma}^2_i=\sum_{t=1}^{T-1} (Y_{i,t} - Y_{i,t+1})^2 /2(T-1)$ is an
unbiased estimator of $\sigma_i^2$.\footnote{To see why $\widehat{\sigma}^2_i$
is unbiased:
$
    \mathbb{E} \left[ \sum_{t=1}^{T-1} (Y_{i,t} - {Y}_{i,t+1})^2 \right]
=
    \mathbb{E} \left[
    \sum_{t=1}^{T-1}U^2_{i,t}+\sum_{t=1}^{T-1}U^2_{i,t+1} \right]
    =
    2(T-1) \sigma_i^2.$}
Since both the numerator and the denominator in~(\ref{optimal-feasible}) can be
negative, 
we propose a better-performing version of feasible weights, obtained by taking the positive part of the numerator
in~(\ref{optimal-feasible}) and then the positive part of the resulting
weights:\footnote{Alternatives such as
just taking the positive part of the weights in~(\ref{optimal-feasible}) or using
the sample covariance between $Y_{i,t}$ and $Y_{i,t-1}$ as an estimator of
$\lambda^2_i$ in the numerator delivered very large errors in
simulations.}
\begin{align}
    W_{i,T}^{IW-O} = \left( \frac{ \left(\sum_{t=1}^T (Y_{i,t}-\mu)^2 /T -
    \sum_{t=1}^{T-1} (Y_{i,t} - Y_{i,t+1})^2
    /2(T-1)\right)^+}{\sum_{t=1}^T (Y_{i,t}-\mu)^2 /T
    -\sum_{t=1}^{T-1} (Y_{i,t} - Y_{i,t+1})^2 /2T}\right)^+,
\end{align}
where $(\cdot)^+$ denotes the positive part.

The short $T$ leads to imprecise estimates of the parameters
in~(\ref{owgen}). 
Our simulations indicate that these weights, even when taking the positive part, can still perform poorly in practice.
This motivates our focus on developing feasible weights that are
robust, specifically MMR-optimal. 

\subsection{Minimax Regret Optimal Weights (IW-MR)}
To derive feasible MMR-optimal weights we shift from
unconditional MSFE to MSFE conditional on the information set at time $T$. The next lemma is the analog of Lemma~\ref{thm:0} for the conditional MSFE.

\begin{lemma}\label{thm:0_conditional}
Consider the forecasts:
\begin{equation}\label{fcstgeneral}
	\widehat{Y}^{IW}_{i,T}=\widehat{Y}^{TS}_{i,T} {W}_{i,T} +
	\widehat{Y}^{Pool}_{i,T} (1-{W}_{i,T}); \;\;\;\;\;\;\;
	\widehat{Y}^{TS}_{i,T}=\bar{Y}_{i,T};  \;\;\;\;\;\;\;
	\widehat{Y}^{Pool}_{i,T}= \mu. 
\end{equation}
Let Assumption~\ref{indep-assumption} hold. The MSFEs conditional on the
information set at time $T$, $\mathcal{Y}_{N,T}$, are
\begin{align*}
    \mathrm{MSFE}(\mathrm{TS}, \theta_i | \mathcal{Y}_{N,T}) &=
        \sigma_{i,T}^2 +\gamma_{i,T}^2, \\
    \mathrm{MSFE}(\mathrm{Pool}, \theta_i | \mathcal{Y}_{N,T}) &=
        \sigma_{i,T}^2 + \kappa^2_{i,T}, \\
    \mathrm{MSFE}(\mathrm{IW}, \theta_i | \mathcal{Y}_{N,T})
    &=
    \sigma_{i,T}^2 +\gamma_{i,T}^2 {W}_{i,T}^2 +\kappa_{i,T}^2
    \left(1 - {W}_{i,T}\right)^2 - 2\delta_{i,T}{W}_{i,T}
    \left(1 - {W}_{i,T}\right),
\end{align*}
where $\sigma_{i,T}^2 = \mathbb{E} \left[ U_{i,T+1}^2 | \mathcal{Y}_{N,T}
\right]$, $\gamma^2_{i,T}=\mathbb{E} \left[ \bar{U}_{i,T}^2 | \mathcal{Y}_{N,T}
\right]$, $\kappa^2_{i,T}=\mathbb{E} \left[ (A_i-\mu)^2 | \mathcal{Y}_{N,T}
\right]$ and $\delta_{i,T}=\mathbb{E} \left[ (A_i-\mu) \bar{U}_{i,T} |
\mathcal{Y}_{N,T} \right]$, with $\bar{U}_{i,T}=T^{-1} \sum_{t=1}^T U_{i,t}$.
\end{lemma}

Here we consider a different type of regret, defined as the
difference between the conditional MSFE for a generic weight $W_{i,T}$ and the
conditional MSFE corresponding to the conditionally optimal weights $W^*_{i,T}$,
obtained by minimizing $\mathrm{MSFE}(\mathrm{IW}, \theta_i | \mathcal{Y}_{N,T})$
derived in Lemma~\ref{thm:0_conditional}:
\begin{align}\label{condoptimalw}
    W_{i,T}^* = \frac{\kappa_{i,T}^2+\delta_{i,T}}{\kappa_{i,T}^2 +
    \gamma_{i,T}^2+2\delta_{i,T}}.
\end{align}

The regret is:
\begin{align}
    R^*(W_{i,T}, \theta_{i,T} | \mathcal{Y}_{N,T})
    &:=
    \mathrm{MSFE}(W_{i,T}, \theta_{i,T} | \mathcal{Y}_{N,T}) -
    \mathrm{MSFE}(W^*_{i,T},\theta_{i,T} | \mathcal{Y}_{N,T})
    \label{condreg} \\
    &= \gamma_{i,T}^2 W_{i,T}^2 + \kappa_{i,T}^2 (1-W_{i,T})^2 -
    2\delta_{i,T} W_{i,T}(1-W_{i,T}) -
    \frac{\gamma_{i,T}^2 \kappa_{i,T}^2 -
    \delta_{i,T}^2}{\gamma_{i,T}^2 + \kappa_{i,T}^2 + 2\delta_{i,T}}
    \nonumber\\
    &= \gamma_{i,T}^2 \left[ W_{i,T}^2 + \zeta_{i,T}^2
    (1-W_{i,T})^2 - 2\rho_{i,T} W_{i,T}(1-W_{i,T}) -
    \frac{\zeta_{i,T}^2 - \rho_{i,T}^2}{1 + \zeta_{i,T}^2 +
    2\rho_{i,T}} \right],
    \nonumber
\end{align}
where
\begin{align}\label{zeta}
    \zeta_{i,T}^2 := \frac{\kappa_{i,T}^2}{ \gamma_{i,T}^2}
    =\frac{\mathbb{E} \left[ (A_i-\mu)^2 | \mathcal{Y}_{N,T} \right]}
    { \mathbb{E} \left[ \bar{U}_{i,T}^2 | \mathcal{Y}_{N,T} \right]},
    \qquad \rho_{i,T} := \frac{\delta_{i,T}}{\gamma_{i,T}^2} =
    \frac{\mathbb{E} \left[ (A_i-\mu) \bar{U}_{i,T} | \mathcal{Y}_{N,T}
    \right]}{ \mathbb{E} \left[ \bar{U}_{i,T}^2 | \mathcal{Y}_{N,T} \right]}.
\end{align}

The following theorem derives the optimal Minimax Regret weight under the assumption that we can bound $\zeta_{i,T}^2$, which can be
interpreted as a ``conditional signal-to-noise ratio''.

\begin{theorem}\label{bounds}
Let Assumption~\ref{indep-assumption} hold, set $\delta_{i,T}=0$, and
assume $\sigma_i^2 > 0$. Suppose that $\zeta_{i,T}^2$ in~(\ref{zeta})
satisfies $\zeta_{i,T}^2 \in [0, \tilde{\zeta}_{i,T}^2]$, where
$\tilde{\zeta}_{i,T}^2 > 0$. Then maximum regret is
\begin{align*}
    \max_{\zeta_{i,T}^2 \,\in\, [0,\,\tilde{\zeta}_{i,T}^2]}
    R^*(W_{i,T}, \theta_i | \mathcal{Y}_{N,T})
    =
    {\gamma_{i,T}^2}
    \max \left[
    {W}_{i,T}^2,\;
    {W}_{i,T}^2  + \tilde{\zeta}_{i,T}^2 \left(1 -
    {W}_{i,T}\right)^2
    -
    \frac{\tilde{\zeta}_{i,T}^2}{\tilde{\zeta}_{i,T}^2 + 1}
    \right],
\end{align*}
with $R^*(W_{i,T}, \theta_i | \mathcal{Y}_{N,T})$ defined as
in~(\ref{condreg}). The weight that minimizes maximum regret over
$W_{i,T} \in [0,1]$ is
\begin{align}
    {W}^{IW-MR}_{i,T}
    =
    1 - \frac{1}{\sqrt{\tilde{\zeta}_{i,T}^2 + 1}}.
\end{align}
\end{theorem}

\begin{remark}[Proof strategy]
Write $s = \zeta_{i,T}^2$ and treat $W_{i,T}$ as fixed. Then the scaled regret
is
\begin{align*}
    h(s)
    :=
    W_{i,T}^2 + s(1-W_{i,T})^2 - \frac{s}{s+1},
    \qquad s \geq 0.
\end{align*}
The first two terms are linear in $s$. Differentiating the nonlinear term twice
gives $\frac{d^2}{ds^2}\left(-\frac{s}{s+1}\right)
    =
    \frac{2}{(s+1)^3}
    >
    0$, for all $s \geq 0,$
so $h''(s)>0$ and the regret is convex in $\zeta_{i,T}^2$. Convexity therefore
implies that the maximum over the compact interval
$[0,\tilde{\zeta}_{i,T}^2]$ is attained at one of the two endpoints. In the
proof of Theorem~\ref{bounds}, the endpoint values coincide only at the unique
crossing point $W_{i,T}=W^{IW-MR}_{i,T}$, which yields uniqueness of the
minimizer. The condition $\tilde{\zeta}_{i,T}^2>0$ is both necessary and
sufficient for the optimal weight to be nontrivial, that is,
$W^{IW-MR}_{i,T}\in(0,1)$. If $\tilde{\zeta}_{i,T}^2=0$, the formula yields the
degenerate solution $W^{IW-MR}_{i,T}=0$.
\end{remark}

\begin{remark}[$\delta_{i,T} = 0$]
The term $\delta_{i,T}$ captures the finite-sample dependence between the RE $A_i$ and the time-series average $\bar U_{i,T}$ conditional on the data;
it therefore reflects the intrinsic difficulty of allowing dependence between
forecast errors and weights in short-$T$ environments. In
Theorem~\ref{bounds}, we set $\delta_{i,T}=0$, which can be motivated as
follows.

First, by the conditional Cauchy--Schwarz inequality,
\begin{align}\label{CauSch}
|\delta_{i,T}| \le \sqrt{\gamma_{i,T}^2 \kappa_{i,T}^2},
\end{align}
so $\delta_{i,T}$ is bounded by the conditional signal and noise
components. As $T\to\infty$, $\gamma_{i,T}^2\to0$, which implies
$\delta_{i,T}\to0$. That is, when $T$ is large, $\delta_{i,T}$ is negligible.
This aligns with the conventional asymptotic arguments in the
forecast-combination literature, where dependence between forecast errors and
weights vanishes in large samples.

Second, for fixed $T$, $\delta_{i,T}=0$ when forecast errors and weights are independent. This holds in the split-sample setting of
Section~\ref{sec: MR}, where weights use data up to $T-1$
and TS uses only $Y_T$. In that setting,
$\mathrm{MSFE}(\mathrm{IW}, \theta_i | \mathcal{Y}_{N,T})$ in
Lemma~\ref{thm:0_conditional} reduces to $\mathrm{MSFE}(\mathrm{IW}, \theta_i | \mathcal{Y}_{N,T-1})$ and does not contain
the last term. In practice, such independence can be ensured by constructing forecasts and weights on non-overlapping time windows, e.g., TS uses the most recent observation---or an average of the most recent---while weights are computed exclusively from lagged data.

Third, $\delta_{i,T}=0$ everywhere is stronger than
necessary for Theorem~\ref{bounds}. Since regret is convex in
$\zeta_{i,T}^2$, the maximum regret occurs at the endpoints, $\zeta_{i,T}^2 = 0$ and $\zeta_{i,T}^2 = \tilde{\zeta}_{i,T}^2$. Thus,
it is only necessary that $\delta_{i,T}=0$ at these endpoints,
not for all values of $\zeta_{i,T}^2$. At the lower endpoint $\zeta_{i,T}^2 = 0$, we necessarily have $\kappa_{i,T}^2
= 0$, and by the Cauchy--Schwarz inequality~(\ref{CauSch}), $\delta_{i,T} = 0$
regardless of $\gamma_{i,T}^2$. Intuitively, when an individual
is close to the mean, the assumption $\delta_{i,T}=0$ is satisfied, independently
of the conditional variance of the error term or of the sample size $T$. At the upper endpoint $\zeta_{i,T}^2 = \tilde{\zeta}_{i,T}^2$, the same
inequality implies $|\delta_{i,T}| \le \gamma_{i,T}^2 \sqrt{\zeta_{i,T}^2}$.
Hence, $\delta_{i,T}$ will still be negligible provided that $\gamma_{i,T}^2$ is
small. This situation arises when an individual is an outlier but has a small
conditional error variance, so that the averaging in $\bar U_{i,T}$ effectively
dampens the idiosyncratic noise. Therefore, $\delta_{i,T}=0$ may fail only in finite samples when an outlier also has a large conditional variance of the error
term, even after averaging across $T$. Theorem~\ref{bounds} can
be interpreted as best describing those weighted estimators whose maximum regret
is attained when either $\gamma_{i,T}^2$ is small and $\kappa_{i,T}^2$ is large,
or vice versa.

In principle one can derive the optimal weights
without imposing $\delta_{i,T}=0$.\footnote{It can be shown that the optimal
weights in the general case are $
    W_{i,T}
    =
    1
    -
\left(
        \rho_{i,T}
        \;\pm\;
        (\lvert \rho_{i,T} + \zeta_{i,T}^2 \rvert)/(
        \sqrt{1 + \zeta_{i,T}^2 + 2\rho_{i,T}})
    \right)/(
        \zeta_{i,T}^2 + 2\rho_{i,T}).$}
However, as shown in footnote (16), the closed-form solution is
complicated and does not lead to a feasible or interpretable
rule. We therefore focus on the empirically relevant case in
which $\delta_{i,T}$ is negligible.
\end{remark}

\medskip
\noindent
\textbf{Feasible IW-MR Weights:}
In practice, the value of the upper bound $\tilde{\zeta}_{i,T}^2$ is uncertain,
but the following heuristic rule can be used to obtain feasible MMR-optimal weights for IW. Assuming $T \geq 2$, we define:
\begin{align}
    \widehat{\tilde{\zeta}_{i,T}^2}
    :=
    \frac{ \max \{ (Y_{i,1}-\mu)^2, \ldots,(Y_{i,T}-\mu)^2 \}}
    { \sum_{t=1}^{T-1} (Y_{i,t} - Y_{i,t+1})^2 /2T(T-1)},
\end{align}
where $\mu$ is either known or approximated by the pooled mean. Intuitively, the
denominator is an unbiased estimator of $\sigma_i^2/T$, which approximates
$\gamma^2_{i,T}=\mathbb{E} \left[ \bar{U}_{i,T}^2 | \mathcal{Y}_{N,T} \right]$
with the unconditional mean $\mathbb{E} \left[ \bar{U}_{i,T}^2
\right]=\sigma^2_i/T$. The numerator serves as a proxy for the upper bound on
$\kappa_{i,T}^2=\mathbb{E} \left[ (A_i-\mu)^2 | \mathcal{Y}_{N,T} \right]$, the
numerator of $\zeta_{i,T}^2$.

Although the construction is heuristic, we may interpret our result as a
minimax-regret optimal rule conditional on $\tilde{\zeta}_{i,T}^2 =
\widehat{\tilde{\zeta}_{i,T}^2}$. This is similar in spirit to partial
identification settings where the outcome variable is known to lie within a
bounded interval (e.g., $Y \in [y_{\min}, y_{\max}]$). When $y_{\min}$ and
$y_{\max}$ are unknown, it is common to use the sample minimum and maximum as
proxies, and interpret the resulting identification region as conditional on these
sample bounds.

These are the weights that perform best in our simulations:\footnote{A similar
performance in simulations and in the empirical applications is obtained by the
following IW-MR rule, based on an alternative unbiased estimator for $\sigma_i^2$: $W_{i,T}^{IW-MR2}= 1 - 1/\left(\sqrt{\frac{ \max \{
    (Y_{i,1}-\mu)^2, \ldots,(Y_{i,T}-\mu)^2 \}}{ \sum_{t=1}^T (Y_{i,t} -
    \bar{Y}_{i,T})^2 /T(T-1)} + 1}\right)$.}
\begin{align} \label{MR}
    W_{i,T}^{IW-MR}= 1 - \frac{1}{\sqrt{\frac{ \max \{
    (Y_{i,1}-\mu)^2, \ldots,(Y_{i,T}-\mu)^2 \}}{ \sum_{t=1}^{T-1}
    (Y_{i,t} - Y_{i,t+1})^2 /2T(T-1)} + 1}}.
\end{align}

\subsection{Inverse MSFE Weights (IW-MSFE)}\label{sec: InverseMSFE}

The weights in this subsection do not rely on the model and the
assumptions, and are thus applicable in more general settings. The weights compare the (in-sample or out-of-sample) MSFE of TS
and Pool. They are analogous to those considered in the forecast combination literature (e.g., \cite{bates1969combination},
\cite{stock1998comparison}), with the difference that the MSFE is computed here
for each individual over a very small time-series. As in \cite{stock1998comparison}, the weights ignore any
correlation between TS and Pool.\footnote{In the time-series
literature, these weights are known to perform well even when the time dimension
is large because of the challenges in estimating correlations precisely; see,
e.g., \cite{stock1998comparison}.}

The in-sample inverse MSFE weights, which perform second-best in our
simulations are:
\begin{align}\label{msfe_is_iw}
    {W}_{i,T}^{IW-MSFE-IS}
    &:=
    \frac{1/\left[\sum_{t=1}^T(Y_{i,t} -
    \widehat{Y}^{TS}_{i,T})^{2}\right]}{1/\left[\sum_{t=1}^T(Y_{i,t} -
    \widehat{Y}^{TS}_{i,T})^{2}\right] + 1/\left[\sum_{t=1}^T(Y_{i,t} -
    \widehat{Y}^{Pool}_{i,T})^{2}\right] }.
\end{align}
The out-of-sample inverse MSFE weights are given by:
\begin{align}
    {W}_{i,T}^{IW-MSFE-OOS}
    &:=
    \frac{(Y_{i,T} - \widehat{Y}^{TS}_{i,T-1})^{-2}}{(Y_{i,T} -
    \widehat{Y}^{TS}_{i,T-1})^{-2} + (Y_{i,T} -
    \widehat{Y}^{Pool}_{i,T-1})^{-2} }.
\end{align}
Here we base ${W}_{i,T}^{OOS}$ only on the out-of-sample forecast
errors at time $T$ corresponding to the TS and Pool forecasts built on data up to $T-1$. Depending on $T$, one could compute the out-of-sample MSFEs using more than just one out-of-sample period, e.g., selecting $P<T$:
\begin{align}
    {W}_{i,T,P}^{IW-MSFE-OOS}
    &:=
    \frac{1/\left[\sum_{t=T-P+1}^T(Y_{i,t} -
    \widehat{Y}^{TS}_{i,t-1})^{2}\right]}{1/\left[\sum_{t=T-P+1}^T(Y_{i,t} -
    \widehat{Y}^{TS}_{i,t-1})^{2}\right] +
    1/\left[\sum_{t=T-P+1}^T(Y_{i,t} -
    \widehat{Y}^{Pool}_{i,t-1})^{2}\right]},
\end{align}
where $\widehat{Y}^{TS}_{i,t-1}$ and $\widehat{Y}^{Pool}_{i,t-1}$ are
respectively the TS and Pool forecasts using data up to time
$t-1$ (either all available data or an arbitrary number of the most recent
observations). Finally, one could consider ``rolling-window''
forecasts, both as the original TS and Pool forecasts and in the computation of
the weights. In this case, both TS and Pool at time $t$ would be based
only on the $R<t$ most recent observations, rather than all available
observations up to time $t$.

\section{Monte Carlo Simulations}\label{sec: MC}

We study the finite-sample performance of feasible IW weights. We then compare IW to JS.

\subsection{Comparing Feasible IW Rules}
Consider one individual (so drop the subscript $i$) observed over 3 periods, with $U_{1},\ldots,U_{3}$ i.i.d. from a
$\mathcal{N}(0, 1)$ and RE $A \sim \mathcal{N}(0,\lambda^2)$, with
$\lambda^2$ taking $50$ equally-spaced values on $[0.001, 2]$.
Repeating the simulation 10{,}000 times approximates the individual
MSFE when forecasting $Y_3$ at $T=2$ under the feasible IW
rules of Section~\ref{sec: MRFeasWeights}, with
$\widehat{Y}_{i,T}^{TS}=\bar{Y}_i$ and $\mu=0$.
Figures~\ref{fig:compareIWsmsfe} and~\ref{fig:compareIWsmmr} respectively
report the MSFEs of each rule relative to IW-MR, and the regret
of each rule, as a function of $\lambda^2$. In Figure~\ref{fig:compareIWsmsfe}, a line above 1 means that the
rule is dominated by IW-MR.

%

\begin{figure}[t]
	\centering
	
	\begin{subfigure}{0.45\textwidth}
		\centering
		\includegraphics[width=\textwidth]{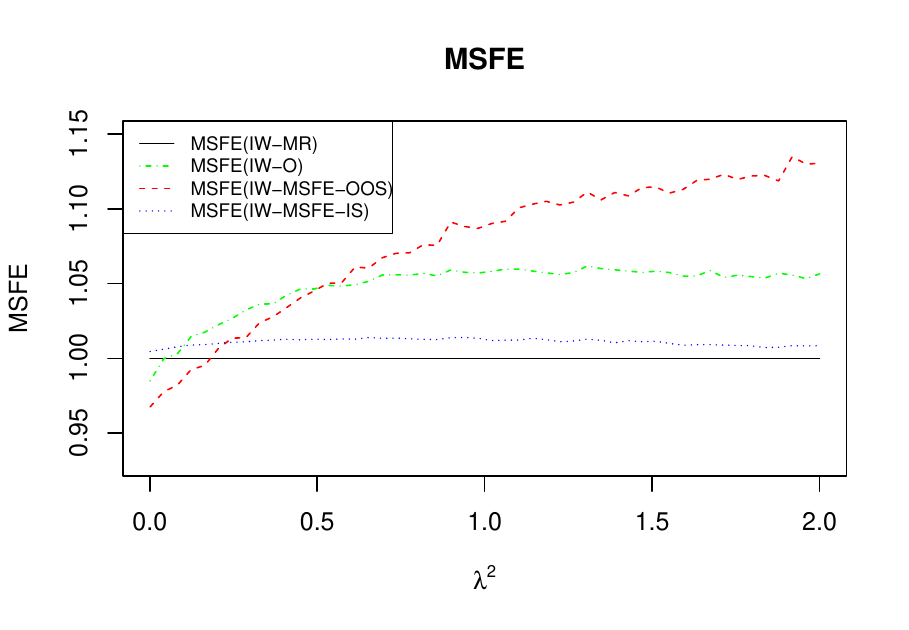}
		\caption{MSFE}
		\label{fig:compareIWsmsfe}
	\end{subfigure}
	\hfill
	\begin{subfigure}{0.45\textwidth}
		\centering
		\includegraphics[width=\textwidth]{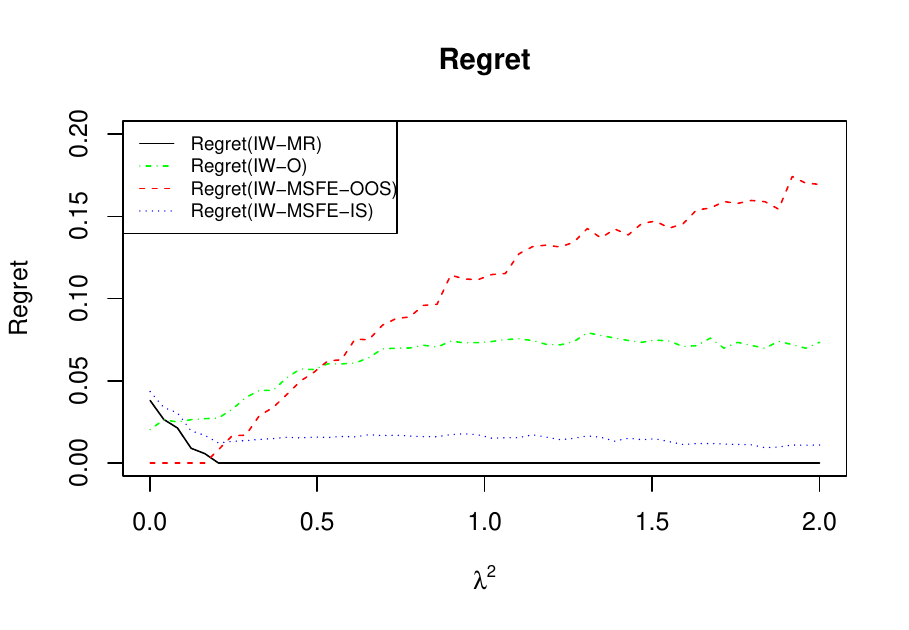}
		\caption{Regret}
		\label{fig:compareIWsmmr}
	\end{subfigure}
	
	\caption{MSFE and regret of alternative feasible IW rules as functions of $\lambda^2$.}
	\label{fig:compareIWs}
	
\end{figure}

Figures~\ref{fig:compareIWsmsfe} and~\ref{fig:compareIWsmmr} show that IW-MR (black solid line) dominates the other feasible rules in both MSFE and regret. 
IW-MSFE-IS (blue dotted line)
is uniformly dominated by IW-MR, although by a small
amount. IW-MSFE-OOS (red dashed line) and IW-O (green dashed-dotted line)
outperform IW-MR when $\lambda^2$ is very low, but perform poorly
over the rest of the parameter space. We therefore take IW-MR as the
preferred rule, closely followed by IW-MSFE-IS.

\subsection{IW vs.\ JS}

We next compare IW-MR with JS. The simulations have two interpretations: they can represent repeated draws of the RE for one individual, so averages across simulations approximate individual MSFE, or different individuals drawn from the same RE distribution, so averages approximate aggregate MSFE. Under the first interpretation, JS is IW with constant weights; under the second, JS is the forecast that
exploits information from the cross-section, in contrast to IW, which
leverages the time-series dimension. The assumption of parameter
homogeneity made by JS is satisfied in all designs below.

\medskip
\noindent
\textbf{Tyranny of the Majority:}
We start by visually illustrating how IW overcomes the ``tyranny of the majority" that affects JS. Henceforth, we focus on the IW-MR rule, which we
saw in the previous section generally outperforms the other feasible rules. We
consider 10{,}000 simulations of outcomes generated as $Y_{t} = A + U_{t}$,
with $t=1,\ldots,3$, $U_{t} \sim \mathcal{N}(0,1)$, independent across $t$. For
the RE we consider the following designs: 1) Normal: $A \sim \mathcal{N}(0,\lambda^2)$, where
	$\lambda^2 \in \{1,3\}$; 2) Laplace: $A \sim \textit{Laplace}(0,1)$, which implies
	mean 0 and variance $\lambda^2=2$; 3) Double Pareto: $A \sim \textit{Double Pareto}(\theta,\beta)$,
	where $\theta=3$ and $\beta=1$ (which implies mean 0 and variance
	$\lambda^2 \approx 1.1$).
These designs correspond to increasing heaviness in the tails of the RE
distribution.

We compare IW-MR as described in Section~\ref{sec: MRFeasWeights} (with
$\widehat{Y}^{Pool}_T=0$) to JS (with estimated weights, as reported after
equation~(\ref{jsrule})). Figure~\ref{fig:tyranny} reports $\Delta SFE$, the difference between the squared
forecast errors of forecasts made at $T=2$ for IW-MR versus JS, for the
different designs. The horizontal axis reports the value of $A$. Negative values therefore favor IW-MR. The figure
illustrates the ``tyranny of the majority": JS tends to make larger
errors than IW-MR for RE in the tails and also near
the center of the distribution. This pattern is not yet visible in panel (a) for
the normal design with low variance, where the cloud appears approximately
symmetric relative to the x-axis, but it is clear in the remaining
panels. For example, in panel (b) (the normal design with larger variance) the
cloud is heart-shaped, showing the superior performance of IW-MR near the center
of the distribution. Panel (c) (the Laplace design) also shows the heart shape
but additionally shows clear tail gains of IW-MR. The
improvement in the tails is starkly evident in panel (d) (the Double Pareto
design), where the cloud has an inverted U-shape. These results illustrate that
what matters for the ``tyranny of the majority" is not only the tail heaviness of
the RE distribution, but its relationship to the variance: panel (b) shows that
the phenomenon is present even when the distribution has thin tails but large
variance. This is intuitive, as both high variance and heavy tails make it
worthwhile to link the shrinkage to the RE (IW) instead of shrinking every
individual by the same amount (JS).

\begin{figure}[t]
	\centering
	
	\scalebox{0.7}{%
		\begin{minipage}{\textwidth}
			\centering
			
			\begin{minipage}{0.45\textwidth}
				\centering
				\includegraphics[width=\linewidth]{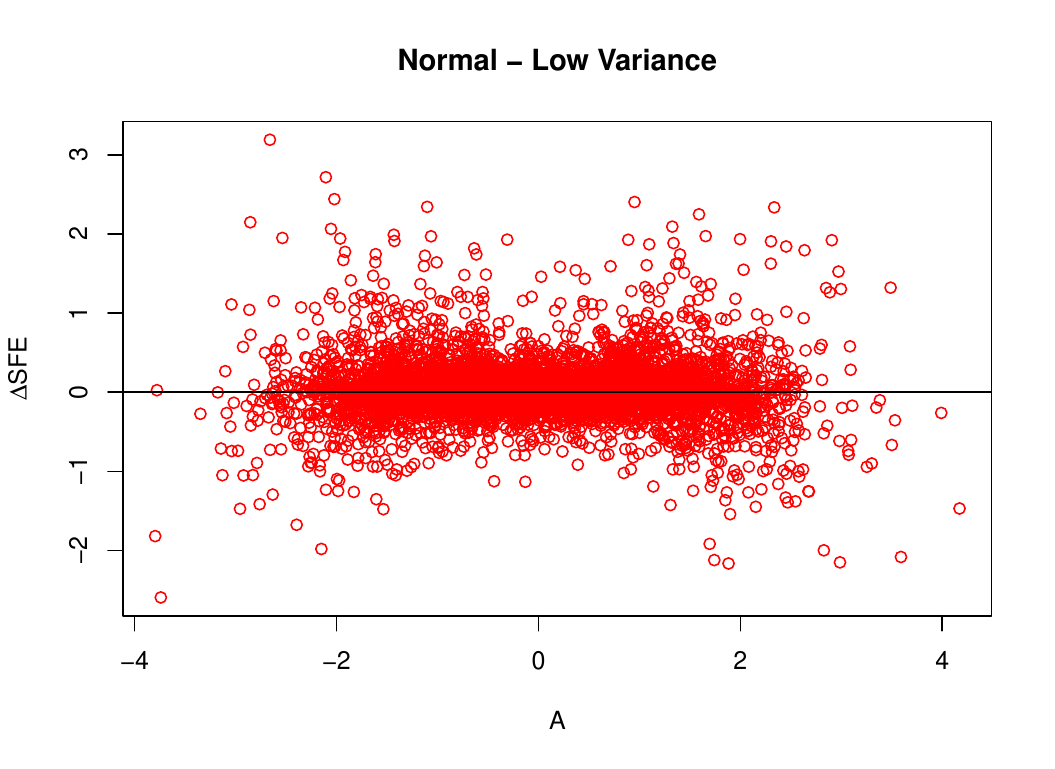}
				\caption*{(a) Normal - Low Variance}
			\end{minipage}
			\hfill
			\begin{minipage}{0.45\textwidth}
				\centering
				\includegraphics[width=\linewidth]{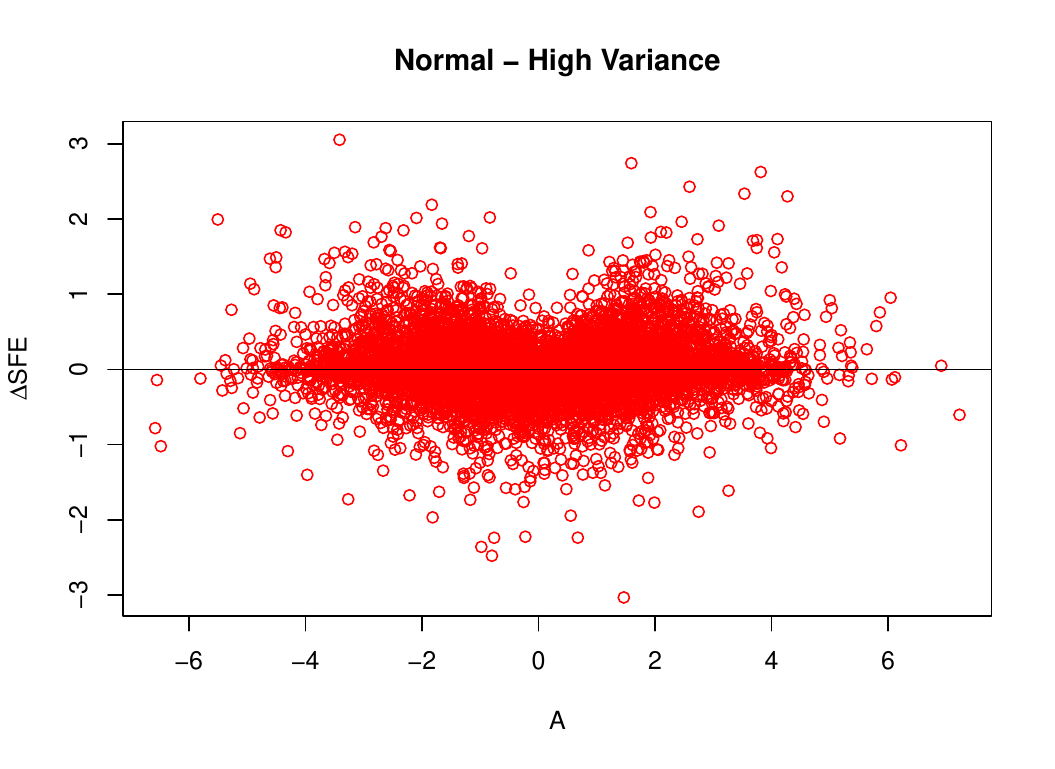}
				\caption*{(b) Normal - High Variance}
			\end{minipage}
			
			\vspace{0.5em}
			
			\begin{minipage}{0.45\textwidth}
				\centering
				\includegraphics[width=\linewidth]{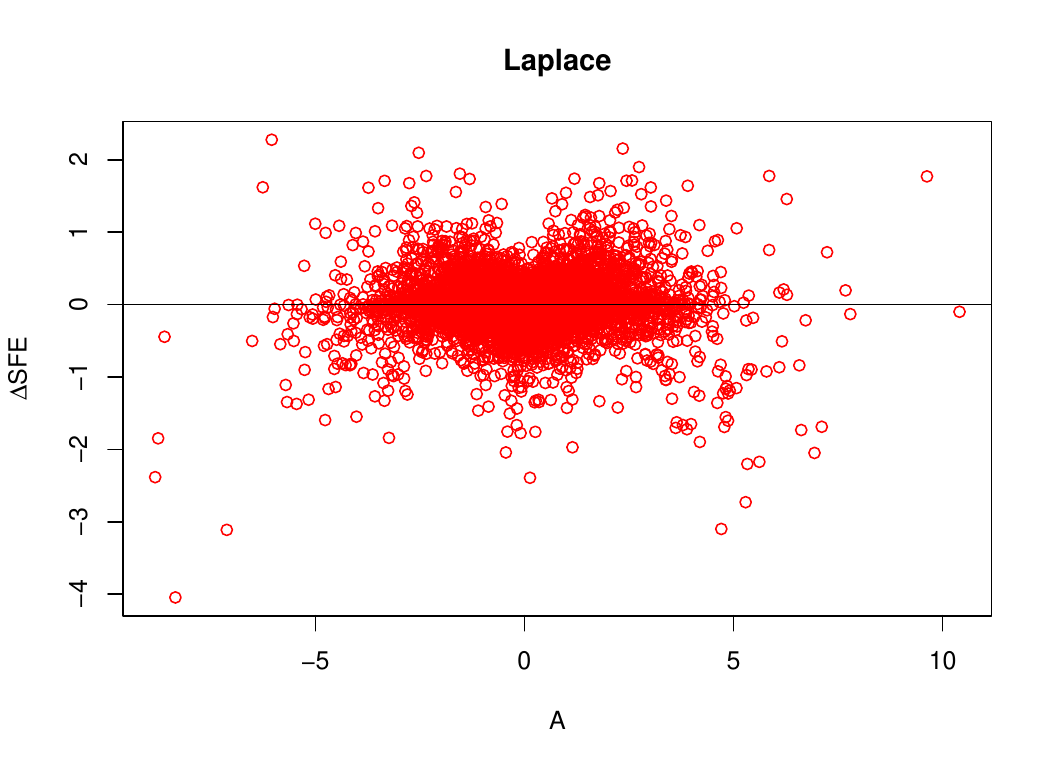}
				\caption*{(c) Laplace}
			\end{minipage}
			\hfill
			\begin{minipage}{0.45\textwidth}
				\centering
				\includegraphics[width=\linewidth]{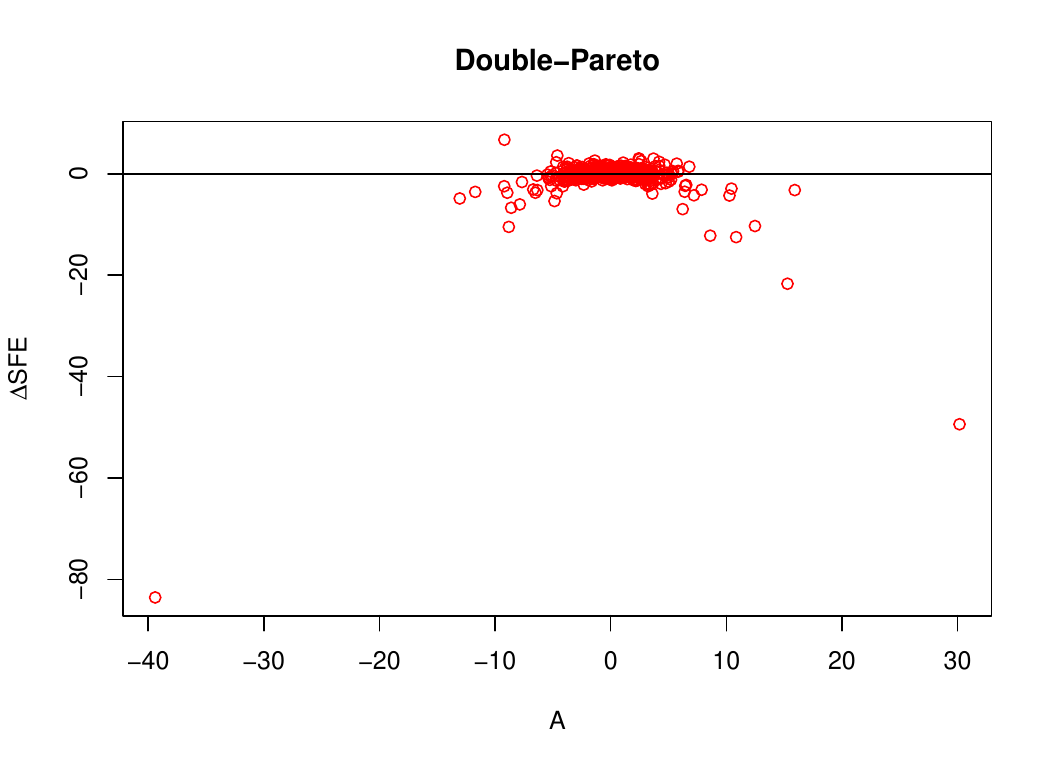}
				\caption*{(d) Double-Pareto}
			\end{minipage}
			
		\end{minipage}%
	}
	
	\caption{Tyranny of the majority across different RE distributions.}
	\label{fig:tyranny}
	
\end{figure}

\medskip
\noindent
\textbf{Aggregate Performance:}
Under the second interpretation, where draws represent different individuals, averaging $\Delta SFE$'s in each panel of Figure~\ref{fig:tyranny} provides a measure of relative aggregate performance. The averages are $0.019$, $0.025$, $-0.005$, and $-0.027$ in panels (a)--(d), respectively. Thus, relative aggregate performance depends on the tail properties of the RE distribution: JS dominates in the normal cases, while IW-MR dominates in the heavy-tailed cases.

\section{Empirical Applications}\label{sec: Application}
We consider two applications of IW, focusing on IW-MR from
Section~\ref{sec: MRFeasWeights}.

\subsection{Estimating and Forecasting Systemic Firm Discrimination}
We use IW to extend the analysis in \cite{kline2022systemic},
assessing the extent to which large U.S.\ employers systemically discriminate
job applicants based on gender. We compare the performance of IW-MR to that of
JS (with estimated weights, as reported after equation~(\ref{jsrule})) and EB
(specifically, the deconvolution estimator of \cite{efron2016empirical},
henceforth Efron).

\medskip
\noindent
\textbf{Data:}
The data are from the experiment in \cite{kline2022systemic}, which sent fictitious applications to jobs posted by 108 of the
largest U.S.\ employers. For each firm, 125 entry-level vacancies were sampled
and 8 job applications with randomized characteristics were sent to each vacancy. Sampling was organized in 5 waves (between October 2019 and
April 2021). Focusing on firms sampled in all waves yields a balanced panel of
$N=72$ firms over $T=5$ waves.\footnote{Accounting for vacancy closures and the
exclusion of some firms from some waves leaves $65{,}400$ applications.}

Applications were sent in pairs, one randomly assigned a distinctively female
name and the other a distinctively male name. For details on the other
observables see \cite{kline2022systemic}. The primary outcome is whether the employer attempted to contact the
applicant within 30 days. The gender contact gap is the
firm-level difference between the male and female contact rate (the ratio of the number of
contacts to the number of received applications).

\medskip
\noindent
\textbf{EB Approach:}
The results in \cite{kline2022systemic} are based on Efron. The approach
considers firm-specific studentized contact gaps, $y_{i,t}=Y_{i,t}/s_i$, where
$Y_{i,t}$ is the contact gap and $s_i$ is the standard deviation of contact gaps
across job applications for firm $i$. These are modelled as
$$y_{i,t}=a_i+u_{i,t}, \;\;\;\;\; u_{i,t}\sim \mathcal{N}(0,1) \;\;\;\;\;
a_i \sim G_{a}, \;\;\;\;\; \text{for } i=1, \ldots, 72.$$
The prior $G_{a}$ belongs to an exponential family, parameterized by a
fifth-order spline. By pooling all five waves, Efron estimates the spline parameters by penalized maximum likelihood and thus obtains the distribution $\hat{G}_{a}$ of studentized contact gaps with density
$\hat{g}_{a}= d\hat{G}_{a}$. One then recovers the distribution $\hat{G}_{A}$
of the RE for the unstudentized contact gaps $Y_{i,t}$ under independence between the RE and $s_i$: the density $\hat{g}_{A} =
d\hat{G}_{A}$ at each $x$ is $\hat{g}_{A}(x)=
\frac{1}{N}\sum_{i=1}^N \frac{1}{s_i} \hat{g}_{a}\!\left(\frac{x}{s_i}\right)$.\footnote{To
perform the deconvolution, the choice of two tuning parameters is required: the
order of the spline and the penalty parameter of the first-step maximum
likelihood procedure. The latter is optimally calibrated to obtain a variance
matching the bias-corrected estimate in Table IV of \cite{kline2022systemic}.}

\medskip
\noindent
\textbf{Estimation and Policy Implications:}
We investigate whether the differences between IW and Efron matter for
estimation and policy. Suppose a counselor advises applicants whether to avoid sending applications to firms
identified as highly discriminatory or discriminatory according to a contact gap threshold, say 0.05 or 0. We thus 
compute $\widehat{\text{Prob}}\left(\hat{Y}^k_{i,T}>0.05\right)$ and
$\widehat{\text{Prob}}\left(\hat{Y}^k_{i,T}>0\right)$ at $T=5$, for $k \in
\{\text{Efron, IW-MR}\}$.\footnote{These probabilities are calculated from one-period-ahead forecasts at $T=5$ of contact gaps based on in-sample
data from all five waves.} We find that the probability of classifying a firm as
highly discriminatory (discriminatory) is 4.29\% (64.29\%) for IW versus
1.43\% (60\%) for Efron, suggesting higher discrimination and thus implying different policy conclusions. 

\medskip
\noindent
\textbf{Forecasting:}
We next compare the forecasting performance of several rules in this setting.
For each wave $T= 3, 4$, we produce one-step-ahead forecasts of (unstudentized)
contact gaps for each firm using TS, the
time-series mean of contact gaps at $T$; Pool, the pooled mean
at $T$; IW-MR; JS; and Efron, obtaining forecasts as posterior mean
estimates of the RE.\footnote{We adapt the code of \cite{kline2022systemic},
used to produce their Figure A13 for assessing the out-of-sample
forecast accuracy of posterior means, to produce forecasts at $T= 3, 4$ using data from
waves $1, \ldots, T$.}
We compare the out-of-sample forecasts from each method $k$,
$\{\hat{Y}^k_{i,T}\}$, to the actual realizations $\{Y_{i,T+1}\}$, for waves
$4, 5$. For each forecasting method $k$ and each firm $i$, the MSFE over the
out-of-sample period is
$    \mathrm{MSFE}(k,i)=\frac{1}{2}\sum_{T=3}^{4}
    (Y_{i,T+1}-\hat{Y}^k_{i,T})^{2}.$

Figure~\ref{fig:Discrimination_sfe} reports the difference $\Delta\mathrm{MSFE}$
between the MSFE of forecasts for IW-MR and those for Efron for each firm. Negative values therefore favor IW-MR. The
horizontal axis shows the value of the gender contact gap at $T=4$. Efron tends to make larger errors for
firms in the right tail or near the center
of the distribution at the time of forecasting, illustrating a possible ``tyranny of the majority"
(which could also be due to misspecification of normality and/or
to the effect of the data-driven choice of the regularization parameter in Efron).

\begin{figure}[t]
    \centering
        \includegraphics[scale=0.4]{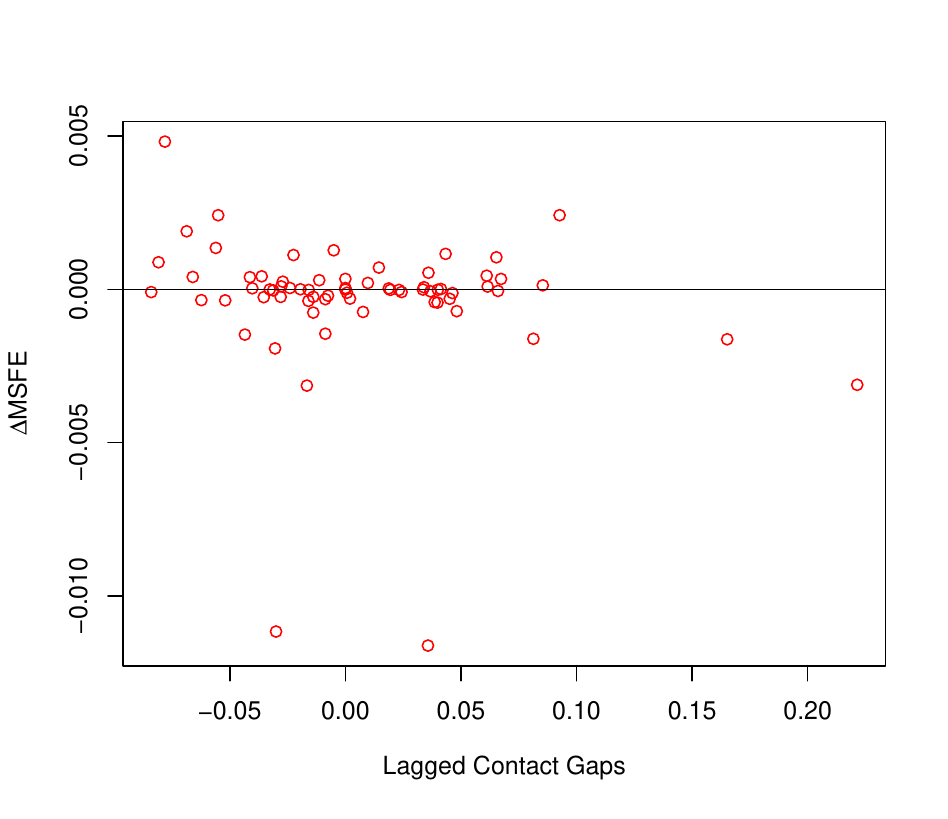}
        \caption{Difference in mean squared forecast errors for each firm
        between IW-MR and Efron.}
        \label{fig:Discrimination_sfe}
\end{figure}

We also evaluate the aggregate performance of the different methods by the average MSFE across $i$ in Table~\ref{Discrimination_MSFE_time_Efron}. The table is descriptive: although IW-MR is not designed to minimize
aggregate loss, it is useful to assess whether a method targeted to individual
performance remains competitive relative to benchmark procedures, such as JS and EB, motivated by aggregate accuracy. Here, IW-MR attains the lowest average MSFE among the methods considered.

\begin{table}[t]
    \caption{Aggregate out-of-sample MSFE}
    \centering
        \scalebox{0.9}{
            \begin{tabular}{ccccc}
                \hline\hline
                TS & Pool & JS & Efron & IW-MR \\
                \hline
                .00297 & .00336 & .00312 & .00314 & .00294 \\
                \hline
            \end{tabular}
        }
    \label{Discrimination_MSFE_time_Efron}
\end{table}

\medskip
\noindent
\textbf{Robustness of IW:}
To study robustness of IW vs.\ Efron, we conduct a subsampling
exercise. We randomly draw $B=1{,}000$ subsamples without replacement, each
subsample $b$ consisting of $n_b=20$ firms. For each method $k \in
\{\text{Efron, IW-MR}\}$, and for each subsample $b$, we calculate the aggregate
out-of-sample RMSFE:
$$\text{RMSFE}_{b, k} = \sqrt{\frac{1}{n_b} \sum_{i=1}^{n_b}
\left(Y_{i, T+1}-\hat{Y}^k_{i,T}\right)^2}$$
at $T=4$, with $\hat{Y}^\text{Efron}_{i,T}$ obtained using the full sample of
firms. In Table~\ref{Discrimination_Subsampling}, we report the minimum,
maximum, mean, median, and 90th percentile of $\text{RMSFE}_{b,k}$ across the
$1{,}000$ subsamples.

\begin{table}[t]
    \caption{Out-of-sample RMSFE across subsamples}
    \centering
        \scalebox{0.9}{
            \begin{tabular}{cccccc}
                \hline\hline
                Method & Min & Max & Mean & Median & 90th pctile \\
                \hline
                Efron  & .02730 & .09312 & .05732 & .05534 & .07458 \\
                IW-MR  & .02662 & .08723 & .05540 & .05497 & .07109 \\
                \hline
            \end{tabular}
        }
    \label{Discrimination_Subsampling}
\end{table}

The results in Table~\ref{Discrimination_Subsampling} demonstrate a sizeable
reduction in the worst-case performance under IW compared to Efron. While the
aggregate performance (Mean or Median) is
slightly better for IW but comparable between the two methods, the difference in
the 90th percentile and Max RMSFE is noticeable: the Max RMSFE is .09312 for
Efron versus .08723 for IW, representing a 6.33\% improvement. This indicates
that Efron is more sensitive to subsample composition, whereas IW effectively reduces the risk of very poor performance (large RMSFE). These
findings highlight the robustness of IW, as evidenced by its
consistently strong performance across the range of sample compositions explored
by this subsampling exercise.

\subsection{Forecasting Earnings}
This section considers an out-of-sample exercise that uses IW to forecasting
earnings residuals. 

\medskip
\noindent
\textbf{Data:}
We use earnings data from the Panel Study of Income Dynamics (PSID) for 1968--1993.\footnote{We use data up
to 1993 because from 1994 a major revision of the survey disrupted the
continuity of PSID files. Moreover, after 1997 the PSID switched from an annual
to a biannual data collection.} Following \cite{meghir2004income}, we select a sample of male workers, heads of
household, aged 24-55 (inclusive). We drop individuals identifying
as Latino, spells of self-employment, observations with zero or top-coded wages, and missing records on race and education. We also require the change in
log earnings to lie between $-3$ and $+5$. We consider earnings
residuals obtained from a first-stage panel regression of log labor income, $\tilde{Y}_{i,t}$, on education, a quadratic
polynomial in age, race, and year dummies. We denote by $Y_{i,t}$ the residuals
from this regression.

\medskip
\noindent
\textbf{Forecasting:} The goal is to obtain individual one-year-ahead forecasts of earnings
residuals.\footnote{Forecasting earnings residuals is of interest since they
	measure individual income risk. Accurate forecasts may matter, for example, for prospective lenders when deciding on loan
	applications.}
We compare the out-of-sample aggregate performance of IW-MR from
Section~\ref{sec: MRFeasWeights}, TS, Pool, or JS.\footnote{We have not included the Efron procedure among the set of
competitors. Its performance is highly sensitive to the fine-tuning of two key
parameters: the order of the spline and the penalty parameter in the first-step
maximum likelihood procedure. Proper calibration of these parameters, especially
on a rolling-window basis, is computationally demanding, and without careful
tuning the method can perform poorly.} We report results for the balanced samples consisting of $N=164$ ($N=790$)
individuals with earnings in all consecutive years for 1968--1993
(1968--1980). We further consider an unbalanced sample built using rolling
windows of $T=3$ time periods of balanced samples of workers (which delivers
sample sizes ranging from $3{,}960$ to $7{,}912$). Forecasts are based on the
model: $ Y_{i,t} = A_i + U_{i,t}$.
We use rolling windows of $T=2$ time periods and compare the out-of-sample
forecasts from each method $k$, $\hat{Y}^k_{i,T}$, where $k \in \{\text{TS,
Pool, IW-MR, JS}\}$, to the actual realizations ${Y}_{i,T+1}$, for
$t=1972,\ldots,1992$, $i=1,\ldots,N$. For each forecasting method $k$ and each
individual $i$, the MSFE over the out-of-sample period is $
    \mathrm{MSFE}(k,i)=\frac{1}{21}\sum_{T=1972}^{1992}
    (Y_{i,T+1}-\hat{Y}^k_{i,T})^{2}.$
    
Table~\ref{psid} reports averages of $\mathrm{MSFE}(k,i)$ across $i$ for each
forecasting method $k$. The table shows that, while TS and JS have lower average
MSFE than Pool, IW attains the lowest average MSFE among the methods considered. Although IW-MR is not designed to
minimize aggregate loss, it remains competitive relative to benchmark forecasts
in terms of aggregate performance.

\begin{table}[t]
    \caption{Aggregate out-of-sample MSFE}
    \centering
        \scalebox{0.9}{
            \begin{tabular}{ccccc}
                \hline\hline
                Sample Size & TS & Pool & IW-MR & JS \\
                \hline
                164           & .075 & .211 & .070 & .072 \\
                794           & .069 & .220 & .067 & .068 \\
                Unbal.\ 4--8000 & .117 & .265 & .108 & .110 \\
                \hline
            \end{tabular}
        }
    \label{psid}
\end{table}

To gain some insight into which individuals benefit more from borrowing strength
(i.e., are given higher weights to Pool by IW), in Figure~\ref{fig:mmr:Clust}
we group the $N=164$ individuals of the balanced sample into deciles of lagged earnings (y-axis) for each year (x-axis). Dot size is proportional to the average weight
attributed by IW-MR to Pool within each year-decile cell.\footnote{We set the size option of the R package \texttt{ggplot} equal
to the mean of weights given to Pool by IW-MR for each quantile.}
Figure~\ref{fig:mmr:Clust} shows that individuals near the median of the
earnings residuals distribution are those who benefit from larger weights to
Pool.

\begin{figure}[t]
    \centering
        \includegraphics[scale=0.4]{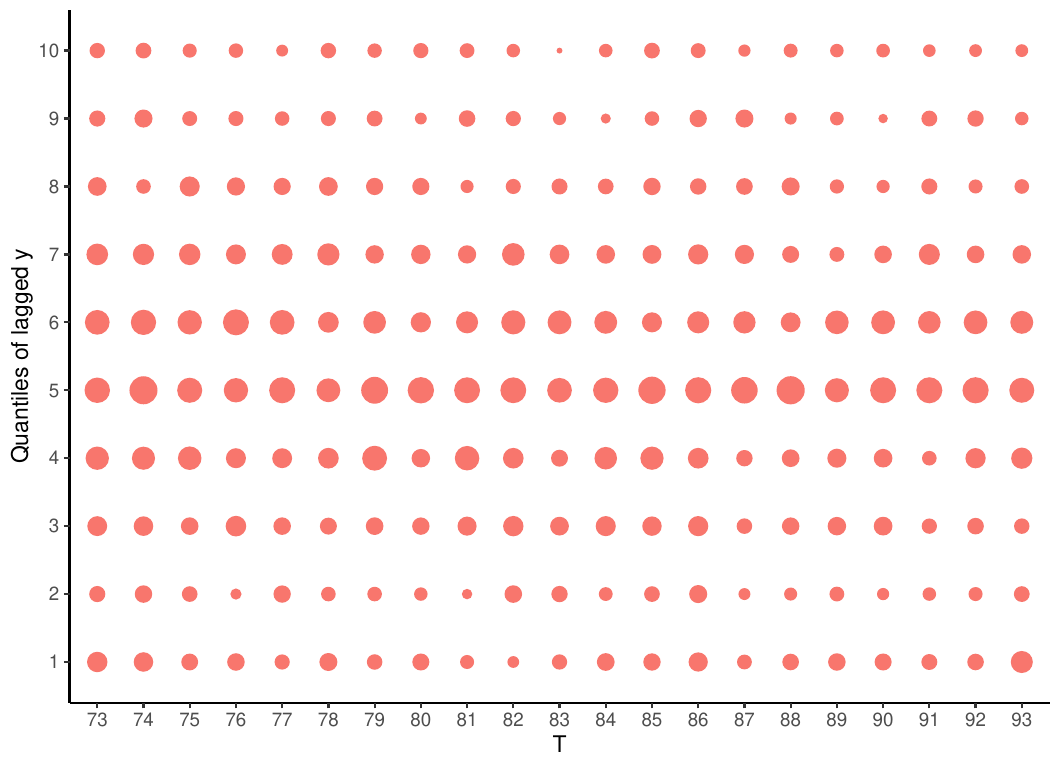}
        \caption{Average weights attributed to Pool by IW-MR by year and
        earnings quantiles.}
        \label{fig:mmr:Clust}
\end{figure}

One interpretation is that the PSID contains enough
unobserved heterogeneity for TS to outperform Pool on average. However, additional gains in
aggregate accuracy can be obtained by using IW, which borrows strength
for individuals near the median of the distribution. 

\section{Conclusion}\label{sec: Conclusions}
Estimating random effects and forecasting individual outcomes with micropanels
are challenging due to the short time dimension, and existing solutions have
shortcomings. We propose a complementary approach that addresses these limitations under minimal assumptions. Our rule shrinks the time-series mean toward the panel
mean, with individual-specific weights calculated from time-series data.
We propose three feasible weights: estimated oracle weights, Minimax
Regret optimal weights, and inverse-MSFE weights. We find that Minimax Regret optimal weights offer superior
performance, closely followed by (in-sample) inverse-MSFE weights.

Our method applies to linear panel data models and value-added models with homogeneous covariate coefficients. 
 The inverse-MSFE weights accommodate more general models, including those with heterogeneous coefficients for covariates, and alternative loss functions. Extending the
derivation of Minimax Regret optimal weights to more general contexts is
inherently more complex, and we reserve this exploration for future research.
A further natural extension is to consider more general classes of shrinkage
estimators, such as combinations of IW with JS.



\begin{singlespace}
	 \bibliographystyle{plainnat}
	\bibliography{References.bib}
\end{singlespace}

\begin{appendix}

\normalsize

\section{Proofs}\label{sec: Proofs}

\begin{proof}[\bf{Proof of Lemma~\ref{thm:0}}]
The MSFEs for TS and Pool are immediate. For IW, the MSFE is: 
\begin{align*}
\mathrm{MSFE}(\mathrm{IW},\theta_i) 
&= 	\mathbb E\left[ \left(Y_{i,T+1} - \widehat{Y}^{IW}_{i,T}\right)^2\right] \\
&= 	\mathbb E[ \left(\left(Y_{i,T+1} - Y_{i,T}\right){W}_{i,T-1}
+ (Y_{i,T+1} - \mu)(1 - {W}_{i,T-1})\right)^2] \\
&= 	\mathbb E[ \left( (U_{i,T+1}-U_{i,T}){W}_{i,T-1}
+(A_i+U_{i,T+1}-\mu)(1-{W}_{i,T-1})\right)^2] \\
&=
\mathbb E[ \left(\left(U_{i,T+1} - U_{i,T}\right){W}_{i,T-1}\right)^2] 
+\mathbb E\!\left[((A_i+U_{i,T+1} - \mu)(1 - {W}_{i,T-1}))^2\right] \\
&\quad
+2\mathbb E\!\left[\left(U_{i,T+1} - U_{i,T}\right){W}_{i,T-1}
((A_i+U_{i,T+1} - \mu)(1 - {W}_{i,T-1}))\right].
\end{align*}

Under Assumption~\ref{indep-assumption}, $W_{i,T-1}$ is a function of
$(Y_{i,1},\ldots,Y_{i,T-1})$ and is independent of $(U_{i,T},
U_{i,T+1})$, though not of $A_i$. 
Using independence of $W_{i,T-1}$ from $(U_{i,T},U_{i,T+1})$, the first term is $\mathbb{E}[(U_{i,T+1}-U_{i,T})^2]\,\mathbb{E}[W_{i,T-1}^2]
= 2\sigma_i^2\,\mathbb{E}[W_{i,T-1}^2]$.
Using independence of $U_{i,T+1}$ from $(A_i,W_{i,T-1})$, the second term is:
$
\mathbb{E}\!\left[(A_i-\mu)^2(1-W_{i,T-1})^2\right]
+ \sigma_i^2\,\mathbb{E}\!\left[(1-W_{i,T-1})^2\right].$ For the cross term:
\begin{align*}
&\mathbb{E}\!\left[(U_{i,T+1}-U_{i,T})(A_i-\mu+U_{i,T+1})
W_{i,T-1}(1-W_{i,T-1})\right] \\
&= \underbrace{\mathbb{E}\!\left[U_{i,T+1}(A_i-\mu)
	W_{i,T-1}(1-W_{i,T-1})\right]}_{=\,0}
+ \sigma_i^2\,\mathbb{E}\!\left[W_{i,T-1}(1-W_{i,T-1})\right] \\
&\quad
- \underbrace{\mathbb{E}\!\left[U_{i,T}(A_i-\mu)
	W_{i,T-1}(1-W_{i,T-1})\right]}_{=\,0}
- \underbrace{\mathbb{E}\!\left[U_{i,T}U_{i,T+1}
	W_{i,T-1}(1-W_{i,T-1})\right]}_{=\,0} \\
&= \sigma_i^2\,\mathbb{E}\!\left[W_{i,T-1}(1-W_{i,T-1})\right].
\end{align*}
The first, third, and fourth terms are zero since $U_{i,T}$ and $U_{i,T+1}$ are independent of $(A_i,W_{i,T-1})$ and from each other, and $\mathbb{E}[U_{i,T}]=\mathbb{E}[U_{i,T+1}]=0$. Combining and using the identity
$\mathbb{E}[W^2]+2\,\mathbb{E}[W(1-W)]+\mathbb{E}[(1-W)^2]=1$
to collect the $\sigma_i^2$ terms,
\begin{align*}
\mathrm{MSFE}(\mathrm{IW}, \theta_i)
&= \sigma_i^2
+ \sigma_i^2\,\mathbb{E}\!\left[{W}_{i,T-1}^2\right]
+ \mathbb{E}\!\left[(A_i - \mu)^2\left(1 - {W}_{i,T-1}\right)^2\right],
\end{align*}
which proves the lemma.
\end{proof}

\begin{lemma}\label{thm:2:IW}
Let $\mathcal{M} = \{\mathrm{TS}, \mathrm{Pool}, \mathrm{IW}\}$. Let
Assumptions~\ref{indep-assumption} and~\ref{key-assumption} hold. Then
$R(\mathrm{IW},\theta_i) \leq \sigma_i^2\nu$ for each $\theta_i \in \Theta$
defined in~\eqref{def:state}. Furthermore, the inequality is strict if either
$0 < {W}_{i,T-1} < 1$ with positive probability or the inequality
in~\eqref{key:regularity:IW} is strict.
\end{lemma}

\begin{proof}[\bf{Proof of Lemma~\ref{thm:2:IW}}]
Invoke Assumption~\ref{key-assumption} to bound the third term of $\mathrm{MSFE}(\mathrm{IW}, \theta_i)$ in 
Lemma~\ref{thm:0}: $\mathbb{E}[(A_i-\mu)^2(1-W_{i,T-1})^2]
\leq \lambda_i^2\,\mathbb{E}[(1-W_{i,T-1})^2]$, giving
\begin{align}\label{key:ineq:IW}
\mathrm{MSFE}(\mathrm{IW}, \theta_i)
&\leq
\sigma_i^2
+ \sigma_i^2\,\mathbb{E}\!\left[{W}_{i,T-1}^2\right]
+ \lambda_i^2\,\mathbb{E}\!\left[\left(1 - {W}_{i,T-1}\right)^2\right].
\end{align}
From Lemma~\ref{thm:0}, $\min_{m\in\mathcal{M}}\mathrm{MSFE}(m,\theta_i)
\leq \sigma_i^2 + \min\{\sigma_i^2,\lambda_i^2\}$, so
\begin{align*}
R(\mathrm{IW},\theta_i)
= \max\!\left\{0,\;
\mathrm{MSFE}(\mathrm{IW},\theta_i) - \sigma_i^2
- \min\{\sigma_i^2,\lambda_i^2\}\right\}.
\end{align*}
The case $\mathrm{MSFE}(\mathrm{IW},\theta_i) < \sigma_i^2 +
\min\{\sigma_i^2,\lambda_i^2\}$ gives $R=0$ trivially, so assume
$\mathrm{MSFE}(\mathrm{IW},\theta_i) \geq \sigma_i^2 +
\min\{\sigma_i^2,\lambda_i^2\}$. Using~\eqref{key:ineq:IW},
\begin{align}\label{key:ineq:IW:more}
\begin{split}
	R(\mathrm{IW},\theta_i)
	&\leq
	\sigma_i^2\,\mathbb{E}\!\left[{W}_{i,T-1}^2\right]
	+ \lambda_i^2\,\mathbb{E}\!\left[\left(1-{W}_{i,T-1}\right)^2\right]
	- \min\{\sigma_i^2,\lambda_i^2\} \\
	&\leq
	\max\{\sigma_i^2,\lambda_i^2\}
	\left(
	\mathbb{E}\!\left[{W}_{i,T-1}^2\right]
	+ \mathbb{E}\!\left[\left(1-{W}_{i,T-1}\right)^2\right]
	\right)
	- \min\{\sigma_i^2,\lambda_i^2\} \\
	&\leq
	\max\{\sigma_i^2,\lambda_i^2\} - \min\{\sigma_i^2,\lambda_i^2\}
	\;\leq\; \sigma_i^2\nu.
\end{split}
\end{align}
The second inequality uses $ax+by\leq\max\{a,b\}(x+y)$ for $a,b,x,y\geq 0$.
The third uses $\mathbb{E}[W^2+(1-W)^2]=\mathbb{E}[1-2W(1-W)]\leq 1$ since
$W(1-W)\geq 0$ for $W\in[0,1]$. The last uses
$|\lambda_i^2-\sigma_i^2|\leq\sigma_i^2\nu$, which follows from
$\lambda_i^2/\sigma_i^2\in[1-\nu,1+\nu]$ on $\Theta$.

The strictness conclusion follows because the inequality in~\eqref{key:ineq:IW}
is strict when~\eqref{key:regularity:IW} is strict, and the third inequality
in~\eqref{key:ineq:IW:more} is strict when $0<W_{i,T-1}<1$ with positive
probability.
\end{proof}

\begin{proof}[\bf{Proof of Theorem~\ref{cor:mmr}}]
By Lemma~\ref{thm:0}, $\min_{m\in\mathcal{M}}\mathrm{MSFE}(m,\theta_i)
\leq\sigma_i^2+\min\{\sigma_i^2,\lambda_i^2\}$.
Thus: 
\begin{align*}
R(\mathrm{TS},\theta_i)
&\geq 2\sigma_i^2 - (\sigma_i^2+\lambda_i^2)
= (\sigma_i^2-\lambda_i^2)\,\mathbb{I}(\sigma_i^2>\lambda_i^2), \\
R(\mathrm{Pool},\theta_i)
&\geq (\lambda_i^2+\sigma_i^2) - 2\sigma_i^2
= (\lambda_i^2-\sigma_i^2)\,\mathbb{I}(\lambda_i^2>\sigma_i^2),
\end{align*}
where $\mathbb{I}\{\cdot\}$ denotes the indicator function. Since
$\lambda_i^2/\sigma_i^2\in[1-\nu,1+\nu]$ on $\Theta$, both maxima equal
$\sigma_i^2\nu$, achieved at $\lambda_i^2=(1-\nu)\sigma_i^2$ and
$\lambda_i^2=(1+\nu)\sigma_i^2$ respectively:
\begin{align}\label{regret-ts-cs}
\max_{\theta_i\in\Theta}R(\mathrm{TS},\theta_i)\geq\sigma_i^2\nu,
\qquad
\max_{\theta_i\in\Theta}R(\mathrm{Pool},\theta_i)\geq\sigma_i^2\nu.
\end{align}
The claim follows from Lemma~\ref{thm:2:IW} and~\eqref{regret-ts-cs}.
\end{proof}

\begin{proof}[\bf{Proof of Theorem~\ref{thm:2:IWeq}}]
Setting $\lambda_i^2=\sigma_i^2$ in the last line of~\eqref{key:ineq:IW:more}
gives $|\lambda_i^2-\sigma_i^2|=0$, so $
\mathrm{MSFE}(\mathrm{IW},\theta_i)
\leq 2\sigma_i^2
= \mathrm{MSFE}(\mathrm{TS},\theta_i)
= \mathrm{MSFE}(\mathrm{Pool},\theta_i).$
The strictness conclusion follows by the same argument as in
Lemma~\ref{thm:2:IW}.
\end{proof}

\begin{proof}[\bf{Proof of Lemma~\ref{thm:0_conditional}}]
Throughout, $W_{i,T}$ is $\mathcal{Y}_{N,T}$-measurable. By
Assumption~\ref{indep-assumption}, $U_{i,T+1}$ is independent of
$(A_i,U_{i,1},\ldots,U_{i,T})$ and hence of $(\mathcal{Y}_{N,T},\bar{U}_{i,T})$,
so
\begin{align}\label{lem:orthogonal}
\mathbb{E}\!\left[U_{i,T+1}\mid\mathcal{Y}_{N,T}\right]=0,
\qquad
\mathbb{E}\!\left[U_{i,T+1}\mid\mathcal{Y}_{N,T},A_i\right]=0.
\end{align}

\noindent\textit{MSFE for TS.} Since $\bar{U}_{i,T}$ is
$\mathcal{Y}_{N,T}$-measurable,
\begin{align*}
\mathbb{E}\!\left[\left(U_{i,T+1}-\bar{U}_{i,T}\right)^2
\Big|\mathcal{Y}_{N,T}\right]
&= \sigma_{i,T}^2
+ \bar{U}_{i,T}^2
- 2\bar{U}_{i,T}
\underbrace{\mathbb{E}\!\left[U_{i,T+1}\mid\mathcal{Y}_{N,T}\right]}_{=\,0}
= \sigma_{i,T}^2 + \gamma_{i,T}^2.
\end{align*}

\noindent\textit{MSFE for Pool.} 
Applying the law of iterated expectations (LIE) with~\eqref{lem:orthogonal},
\begin{align*}
\mathbb{E}\!\left[(A_i-\mu+U_{i,T+1})^2\mid\mathcal{Y}_{N,T}\right]
&= \mathbb{E}\!\left[(A_i-\mu)^2\mid\mathcal{Y}_{N,T}\right]
+ \mathbb{E}\!\left[U_{i,T+1}^2\mid\mathcal{Y}_{N,T}\right] \\
&\quad
+ 2\,\mathbb{E}\!\left[(A_i-\mu)
\underbrace{\mathbb{E}\!\left[U_{i,T+1}\mid\mathcal{Y}_{N,T},A_i\right]}
_{=\,0}
\,\Big|\,\mathcal{Y}_{N,T}\right]
= \kappa_{i,T}^2 + \sigma_{i,T}^2.
\end{align*}

\noindent\textit{MSFE for IW.} Since $W_{i,T}$ is
$\mathcal{Y}_{N,T}$-measurable,
\begin{align*}
\mathrm{MSFE}(\mathrm{IW},\theta_i\mid\mathcal{Y}_{N,T})
&= (\sigma_{i,T}^2+\gamma_{i,T}^2)\,W_{i,T}^2
+ (\kappa_{i,T}^2+\sigma_{i,T}^2)(1-W_{i,T})^2 \\
&\quad
+ 2\,\mathbb{E}\!\left[(U_{i,T+1}-\bar{U}_{i,T})
(A_i-\mu+U_{i,T+1})\mid\mathcal{Y}_{N,T}\right]
W_{i,T}(1-W_{i,T}).
\end{align*}
Expanding the cross term,
\begin{align*}
&\mathbb{E}\!\left[(U_{i,T+1}-\bar{U}_{i,T})(A_i-\mu+U_{i,T+1})
\mid\mathcal{Y}_{N,T}\right] =
\underbrace{\mathbb{E}\!\left[U_{i,T+1}(A_i-\mu)
	\mid\mathcal{Y}_{N,T}\right]}_{=\,0}
+
\underbrace{\mathbb{E}\!\left[U_{i,T+1}^2
	\mid\mathcal{Y}_{N,T}\right]}_{=\,\sigma_{i,T}^2} \\
&  -
\underbrace{\mathbb{E}\!\left[\bar{U}_{i,T}(A_i-\mu)
	\mid\mathcal{Y}_{N,T}\right]}_{=\,\delta_{i,T}}
-
\underbrace{\mathbb{E}\!\left[\bar{U}_{i,T}U_{i,T+1}
	\mid\mathcal{Y}_{N,T}\right]}_{=\,0} = \sigma_{i,T}^2 - \delta_{i,T},
\end{align*}
where the first term vanishes by LIE and~\eqref{lem:orthogonal}, and the fourth by independence of
$U_{i,T+1}$ from $\bar{U}_{i,T}$ and~\eqref{lem:orthogonal}. Using $W^2+(1-W)^2+2W(1-W)=1$ to collect the $\sigma_{i,T}^2$ terms,
\begin{align*}
\mathrm{MSFE}(\mathrm{IW},\theta_i\mid\mathcal{Y}_{N,T})
&= \sigma_{i,T}^2
+ \gamma_{i,T}^2\,W_{i,T}^2
+ \kappa_{i,T}^2(1-W_{i,T})^2
- 2\delta_{i,T}\,W_{i,T}(1-W_{i,T}),
\end{align*}
which proves the lemma.
\end{proof}

\begin{proof}[\bf{Proof of Theorem~\ref{bounds}}]
Since $\sigma_i^2>0$, $\gamma_{i,T}^2=\mathbb{E}[\bar{U}_{i,T}^2\mid
\mathcal{Y}_{N,T}]>0$, 
so minimizing maximum regret over $W_{i,T}\in[0,1]$ is equivalent to
minimizing the scaled maximum regret
\begin{align*}
g(W_{i,T})
:= \max\!\left[
{W}_{i,T}^2,\;
{W}_{i,T}^2 + \tilde{\zeta}_{i,T}^2(1-{W}_{i,T})^2
- \frac{\tilde{\zeta}_{i,T}^2}{\tilde{\zeta}_{i,T}^2+1}
\right].
\end{align*}
Equating the two arguments (using $\tilde{\zeta}_{i,T}^2>0$) gives
the candidate minimizer
\begin{align*}
\widetilde{W}_{i,T}
:= 1 - \frac{1}{\sqrt{\tilde{\zeta}_{i,T}^2+1}}.
\end{align*}
At $\widetilde{W}_{i,T}$, both arguments equal $\widetilde{W}_{i,T}^2$, so
$g(\widetilde{W}_{i,T})=\widetilde{W}_{i,T}^2$.

\noindent\textit{Case 1: $W_{i,T}>\widetilde{W}_{i,T}$.}
Since $W_{i,T}\in[0,1]$, $1-W_{i,T}\geq 0$, 
giving $
\tilde{\zeta}_{i,T}^2(1-W_{i,T})^2
< \frac{\tilde{\zeta}_{i,T}^2}{\tilde{\zeta}_{i,T}^2+1}.$
The second argument is strictly less than $W_{i,T}^2$, so
$g(W_{i,T})=W_{i,T}^2>\widetilde{W}_{i,T}^2=g(\widetilde{W}_{i,T})$.

\noindent\textit{Case 2: $W_{i,T}<\widetilde{W}_{i,T}$.}
Then $(1-W_{i,T})^2>1/(\tilde{\zeta}_{i,T}^2+1)$, so the second argument
exceeds $W_{i,T}^2$ and $g(W_{i,T})=\widetilde{W}_{i,T}^2+f(W_{i,T})$,
where
$
f(W)
:= W^2 + \tilde{\zeta}_{i,T}^2(1-W)^2
- \frac{\tilde{\zeta}_{i,T}^2}{\tilde{\zeta}_{i,T}^2+1}
- \widetilde{W}_{i,T}^2.$
The function $f$ is convex with minimizer
$W^*:=\tilde{\zeta}_{i,T}^2/(\tilde{\zeta}_{i,T}^2+1)$.
Setting $x:=\sqrt{\tilde{\zeta}_{i,T}^2+1}>1$, we have $W^*=1-x^{-2}$ and
$\widetilde{W}_{i,T}=1-x^{-1}$, so
$W^* - \widetilde{W}_{i,T} = x^{-1} - x^{-2} = x^{-1}(1-x^{-1}) > 0,$
meaning $f$ is strictly decreasing on $(-\infty,\widetilde{W}_{i,T}]$.
Since $(1-\widetilde{W}_{i,T})^2=1/(\tilde{\zeta}_{i,T}^2+1)$ by
construction, direct substitution gives
$
f(\widetilde{W}_{i,T})
= \frac{\tilde{\zeta}_{i,T}^2}{\tilde{\zeta}_{i,T}^2+1}
- \frac{\tilde{\zeta}_{i,T}^2}{\tilde{\zeta}_{i,T}^2+1}
= 0,$
so $f(W)>0$ for all $W<\widetilde{W}_{i,T}$, giving
$g(W_{i,T})>\widetilde{W}_{i,T}^2=g(\widetilde{W}_{i,T})$.

Both cases give $g(W_{i,T})>g(\widetilde{W}_{i,T})$, so the unique minimizer
is $\widetilde{W}_{i,T}=W^{IW-MR}_{i,T}$.
\end{proof}

\section{IW for Estimation of RE}\label{sec: MSE}

Rather than focusing on the forecasting problem discussed in the body of the
paper, in this appendix we consider the problem of estimating the RE $A_i$.

\subsection{The Model}
The model is:
\begin{align}\label{modelMSE}
Y_{i,t} = A_i + U_{i,t}, \;\;\; i=1,\ldots,N; \;\;\; t=1,\ldots,T,
\end{align}
where $A_i \sim (0, \lambda_i^2)$ and $U_{i,t} \sim (0, \sigma_i^2)$.
$A_i, U_{i,1},\ldots,U_{i,T}$ are random variables, whereas $\lambda_i^2$
and $\sigma_i^2$ are parameters. We take the frequentist approach.

\begin{assumption}[Independence]\label{indep-assumptionMSE}
$A_i, U_{i,1},\ldots,U_{i,T}$ are mutually independent.
\end{assumption}

\subsection{Optimality of IW}\label{FeasibleIWMSE}
In this section we show conditions under which IW is Minimax Regret optimal
relative to the time-series estimator and the common mean, in a simplified
setting where the weights and the time-series estimator are independent.
Suppose $T=2$, and set $\mu=0$ (consistent with $A_i\sim(0,\lambda_i^2)$).
Consider the following estimators:
\begin{align}
\text{Time series (TS):}\quad
\widehat{Y}^{TS}_{i,2} &= Y_{i,2},  \\
\text{Common mean (Pool):}\quad
\widehat{Y}^{Pool}_{i,2} &= 0,  \\
\text{Shrinkage (IW):}\quad
\widehat{Y}^{IW}_{i,2}
&= Y_{i,2}\,W_{i,1}.
\end{align}

The next lemma derives the MSEs of TS, Pool, and IW when the estimand is $A_i$.
\begin{lemma}\label{thm:0MSE}
Consider the three estimators above. Then under
Assumption~\ref{indep-assumptionMSE},
\begin{align*}
\mathrm{MSE}(\mathrm{TS}, \theta_i) &= \sigma_i^2, \\
\mathrm{MSE}(\mathrm{Pool}, \theta_i) &= \lambda_i^2, \\
\mathrm{MSE}(\mathrm{IW}, \theta_i)
&=
\sigma_i^2\,\mathbb{E}\!\left[{W}_{i,1}^2\right]
+ \mathbb{E}\!\left[A_i^2\left(1-{W}_{i,1}\right)^2\right].
\end{align*}
\end{lemma}

Lemma~\ref{thm:0MSE} suggests that the trade-off between TS and Pool depends
on the ``signal-to-noise'' ratio $\lambda_i^2/\sigma_i^2$: Pool dominates
when the ratio is less than 1 and TS dominates when it is greater than 1.

Let $\mathcal{M}$ include $\mathrm{TS}$, $\mathrm{Pool}$, and $\mathrm{IW}$.
We define regret as
\begin{align}\label{def:regret:1MSE}
R(m,\theta_i) := \mathrm{MSE}(m, \theta_i)
- \min_{h \in \mathcal{M}} \mathrm{MSE}(h, \theta_i).
\end{align}
The Minimax Regret (MMR) criterion selects the estimator $m$ that minimizes
$\max_{\theta_i\in\Theta}R(m,\theta_i)$, where $\Theta$ is the parameter
space.

\begin{assumption}[Individual Weight]\label{key-assumptionMSE}
The individual weight ${W}_{i,1}$ satisfies $0\leq{W}_{i,1}\leq 1$ and
\begin{align}\label{key:regularity:IWMSE}
\mathbb{E}\!\left[A_i^2\left(1-{W}_{i,1}\right)^2\right]
\leq \mathbb{E}\!\left[A_i^2\right]
\mathbb{E}\!\left[\left(1-{W}_{i,1}\right)^2\right].
\end{align}
\end{assumption}

\subsection{Minimax Regret Optimality of IW}

We restrict attention to the parameter space where the signal-to-noise ratio
$\lambda_i^2/\sigma_i^2$ ranges from $1-\nu$ to $1+\nu$ for some
$0\leq\nu<1$:
\begin{align}\label{def:stateMSE}
\Theta = \Theta(\nu)
:= \{(\sigma_i^2, \lambda_i^2)\in\mathbb{R}_{+}^2:
1-\nu \leq \lambda_i^2/\sigma_i^2 \leq 1+\nu\}.
\end{align}

\begin{theorem}\label{cor:mmrMSE}
Let Assumptions~\ref{indep-assumptionMSE} and~\ref{key-assumptionMSE} hold.
Then,
\begin{align*}
\max_{\theta_i\in\Theta} R(\mathrm{IW},\theta_i)
\leq \min\!\left\{
\max_{\theta_i\in\Theta} R(\mathrm{TS},\theta_i),\;
\max_{\theta_i\in\Theta} R(\mathrm{Pool},\theta_i)
\right\},
\end{align*}
where $\Theta$ is defined in~\eqref{def:stateMSE}. Furthermore, the
inequality is strict if either $0<{W}_{i,1}<1$ with positive probability
or the inequality in~\eqref{key:regularity:IWMSE} is strict.
\end{theorem}

\subsection{Proofs}\label{sec: ProofsMSE}

\begin{proof}[\bf{Proof of Lemma~\ref{thm:0MSE}}]
The MSEs for TS and Pool are immediate. For IW, write
$A_i - \widehat{Y}^{IW}_{i,2} = A_i(1-W_{i,1}) - U_{i,2}W_{i,1}$.
Squaring, taking expectations, and noting that the cross term vanishes
since $U_{i,2}$ is independent of $(A_i,W_{i,1})$ and $\mathbb{E}[U_{i,2}]=0$,
\begin{align*}
\mathrm{MSE}(\mathrm{IW},\theta_i)
&= \mathbb{E}\!\left[U_{i,2}^2\right]\mathbb{E}\!\left[W_{i,1}^2\right]
+ \mathbb{E}\!\left[A_i^2(1-W_{i,1})^2\right]
= \sigma_i^2\,\mathbb{E}\!\left[W_{i,1}^2\right]
+ \mathbb{E}\!\left[A_i^2(1-W_{i,1})^2\right],
\end{align*}
which proves the lemma.
\end{proof}

Lemma~\ref{thm:2:IWMSE} is used to prove Theorem~\ref{cor:mmrMSE}.
\begin{lemma}\label{thm:2:IWMSE}
Let $\mathcal{M} = \{\mathrm{TS}, \mathrm{Pool}, \mathrm{IW}\}$. Let
Assumptions~\ref{indep-assumptionMSE} and~\ref{key-assumptionMSE} hold.
Then $R(\mathrm{IW},\theta_i)\leq\sigma_i^2\nu$ for each $\theta_i\in\Theta$
defined in~\eqref{def:stateMSE}. Furthermore, the inequality is strict if
either $0<{W}_{i,1}<1$ with positive probability or the inequality
in~\eqref{key:regularity:IWMSE} is strict.
\end{lemma}

\begin{proof}[\bf{Proof of Lemma~\ref{thm:2:IWMSE}}]
Invoke Assumption~\ref{key-assumptionMSE} to bound the second term of
Lemma~\ref{thm:0MSE}: $\mathbb{E}[A_i^2(1-W_{i,1})^2]
\leq\lambda_i^2\,\mathbb{E}[(1-W_{i,1})^2]$, giving
\begin{align}\label{key:ineq:IWMSE}
\mathrm{MSE}(\mathrm{IW},\theta_i)
\leq
\sigma_i^2\,\mathbb{E}\!\left[{W}_{i,1}^2\right]
+ \lambda_i^2\,\mathbb{E}\!\left[\left(1-{W}_{i,1}\right)^2\right].
\end{align}
Since $\mathrm{MSE}(\mathrm{TS},\theta_i)=\sigma_i^2$ and
$\mathrm{MSE}(\mathrm{Pool},\theta_i)=\lambda_i^2$, we have
$\min_{m\in\mathcal{M}}\mathrm{MSE}(m,\theta_i)\leq\min\{\sigma_i^2,\lambda_i^2\}$,
so
$R(\mathrm{IW},\theta_i)
=\max\{0,\,\mathrm{MSE}(\mathrm{IW},\theta_i)-\min\{\sigma_i^2,\lambda_i^2\}\}$.
The case $\mathrm{MSE}(\mathrm{IW},\theta_i)<\min\{\sigma_i^2,\lambda_i^2\}$
gives $R=0$ trivially, so assume
$\mathrm{MSE}(\mathrm{IW},\theta_i)\geq\min\{\sigma_i^2,\lambda_i^2\}$.
Using~\eqref{key:ineq:IWMSE},
\begin{align}\label{key:ineq:IW:moreMSE}
\begin{split}
	R(\mathrm{IW},\theta_i)
	&\leq
	\sigma_i^2\,\mathbb{E}\!\left[{W}_{i,1}^2\right]
	+ \lambda_i^2\,\mathbb{E}\!\left[\left(1-{W}_{i,1}\right)^2\right]
	- \min\{\sigma_i^2,\lambda_i^2\} \\
	&\leq
	\max\{\sigma_i^2,\lambda_i^2\}
	\!\left(
	\mathbb{E}\!\left[{W}_{i,1}^2\right]
	+ \mathbb{E}\!\left[\left(1-{W}_{i,1}\right)^2\right]
	\right)
	- \min\{\sigma_i^2,\lambda_i^2\} \\
	&\leq
	\max\{\sigma_i^2,\lambda_i^2\} - \min\{\sigma_i^2,\lambda_i^2\}
	\;\leq\; \sigma_i^2\nu,
\end{split}
\end{align}
where $\mathbb{I}\{\cdot\}$ denotes the indicator function. The second
inequality uses $ax+by\leq\max\{a,b\}(x+y)$ for $a,b,x,y\geq 0$. The
third uses $\mathbb{E}[W^2+(1-W)^2]\leq 1$ for $W\in[0,1]$. The last uses
$|\lambda_i^2-\sigma_i^2|\leq\sigma_i^2\nu$ from the definition of $\Theta$
in~\eqref{def:stateMSE}.

The strictness conclusion follows because the inequality in~\eqref{key:ineq:IWMSE}
is strict when~\eqref{key:regularity:IWMSE} is strict, and the third
inequality in~\eqref{key:ineq:IW:moreMSE} is strict when $0<W_{i,1}<1$
with positive probability.
\end{proof}

\section{Feasible Weights for IW in a Split-Sample Setting}
\label{sec: MRFeasWeightssimplelagged}

Here we describe three types of feasible weights obtained in the same
split-sample setting as that used to derive the theoretical results in
Section~\ref{sec: MR}.

\subsection{Estimated Oracle Weights (IW-O)}\label{FeasibleIW_simplelagged}
The first set of feasible weights are based on the oracle weights that
minimize the individual MSFE,
$\mathrm{MSFE}(\widehat{Y}^{IW}_{i,T})
=\mathbb{E}\!\left[\left(Y_{i,T+1}-\widehat{Y}^{IW}_{i,T}\right)^2\right],$
which are functions of the individual variance parameters:\footnote{These
oracle weights follow for example from equation (9) in Chapter 4 of
\cite{timmermann2006forecast}, using the fact that the joint distribution
of $Y_{i,T+1}$ and ${Y}_{i,T}$ is
$$\begin{pmatrix} Y_{i,T+1}\\ {Y}_{i,T} \end{pmatrix}
\sim \left(
\begin{pmatrix} 0\\ 0 \end{pmatrix},
\begin{pmatrix}
\lambda_i^2+\sigma_i^2 & \lambda_i^2\\
\lambda_i^2 & \lambda_i^2+\sigma_i^2
\end{pmatrix}
\right),$$
which gives the optimal weight on ${Y}_{i,T}$ as the inverse of the variance of the forecast times the covariance between the
outcome and the forecast. The linear combination with $W^{o}_{i}$ as weight
could also be obtained as the ``best linear rule'' in equation
(9.4), of \cite{efron1973stein} with
${Y}_{i,T}|A_i\sim(A_i,\sigma_i^2)$ and $A_i\sim(0,\lambda_i^2)$.}
\begin{align}\label{ow}
W^{o}_{i} = \frac{\lambda_i^2}{\lambda_i^2 + \sigma_i^2}.
\end{align}
Estimated oracle weights at time $T-1$ can be obtained as
\begin{align}\label{optimal}
W_{i,T-1}^{IW-O}
= \frac{\sum_{t=1}^{T-1}(Y_{i,t}-\mu)^2/(T-1)
- \sum_{t=1}^{T-2}(Y_{i,t}-Y_{i,t+1})^2/[2(T-2)]}
{\sum_{t=1}^{T-1}(Y_{i,t}-\mu)^2/(T-1)},
\end{align}
using the facts that $\sum_{t=1}^{T-1}(Y_{i,t}-\mu)^2/(T-1)$ is an unbiased
estimator of $\lambda_i^2+\sigma_i^2$, and that
$\widehat{\sigma}^2_i=\sum_{t=1}^{T-2}(Y_{i,t}-Y_{i,t+1})^2/[2(T-2)]$ is
an unbiased estimator of $\sigma_i^2$.\footnote{To see that
$\widehat{\sigma}^2_i$ is an unbiased estimator of $\sigma_i^2$, note that:
\begin{align*}
\mathbb{E}\!\left[\sum_{t=1}^{T-2}(Y_{i,t}-{Y}_{i,t+1})^2\right]
&= \mathbb{E}\!\left[\sum_{t=1}^{T-2}U^2_{i,t}
+\sum_{t=1}^{T-2}U^2_{i,t+1}\right]
= 2(T-2)\sigma_i^2.
\end{align*}}

The short time dimension leads to imprecise estimates of these parameters. 
While taking the positive part---as we do for the general weights
reported in Section~\ref{sec: MRFeasWeights}---partially alleviates this
issue, our simulations indicate that these weights still perform poorly in
practice. These considerations motivate our focus on developing feasible
weights that are robust.

\subsection{Minimax Regret Optimal Weights (IW-MR)}
In order to obtain feasible Minimax Regret optimal weights, we shift from
unconditional MSFE to MSFE conditional on the information set at time $T-1$.
The following lemma is the analog of Lemma~\ref{thm:0} for the conditional
MSFE.

\begin{lemma}\label{thm:0_conditionalsimple}
Consider the forecasts in~(\ref{IWsimple}). Let
Assumption~\ref{indep-assumption} hold. Since $W_{i,T-1}$ depends only on
$\{Y_{i,1},\ldots,Y_{i,T-1}\}$ and the time-series forecast uses only
$Y_{i,T}$, the weights and forecast errors are independent, so
$\delta_{i,T-1}=0$ exactly. The MSFEs conditional on
$\{Y_{i,1},\ldots,Y_{i,T-1}\}$ are
\begin{align*}
\mathrm{MSFE}(\mathrm{TS},\theta_i\mid Y_{i,1},\ldots,Y_{i,T-1})
&= 2\sigma_i^2, \\
\mathrm{MSFE}(\mathrm{Pool},\theta_i\mid Y_{i,1},\ldots,Y_{i,T-1})
&= \kappa_{i,T-1}^2 + \sigma_i^2, \\
\mathrm{MSFE}(\mathrm{IW},\theta_i\mid Y_{i,1},\ldots,Y_{i,T-1})
&= \sigma_i^2(1+{W}_{i,T-1}^2)
+ \kappa_{i,T-1}^2(1-{W}_{i,T-1})^2,
\end{align*}
where
\begin{align}\label{kappadef}
\kappa_{i,T-1}^2
:= \mathbb{E}\!\left[(A_i-\mu)^2\mid Y_{i,1},\ldots,Y_{i,T-1}\right].
\end{align}
\end{lemma}

Differentiating $\mathrm{MSFE}(\mathrm{IW},\theta_i\mid
Y_{i,1},\ldots,Y_{i,T-1})$ with respect to $W_{i,T-1}$ and setting the
derivative to zero, the conditionally optimal weights are
\begin{align}\label{cow}
{W}_{i,T-1}^* = \frac{\kappa_{i,T-1}^2}{\kappa_{i,T-1}^2+\sigma_i^2}.
\end{align}

We define regret as the difference between the conditional MSFE for a
generic weight $W_{i,T-1}$ and the conditional MSFE at the optimal weight
$W^*_{i,T-1}$ in~\eqref{cow}:
\begin{align}\label{condregsimple}
\begin{split}
&R^*(W_{i,T-1},\theta_i\mid Y_{i,1},\ldots,Y_{i,T-1}) \\
&:= \mathrm{MSFE}(W_{i,T-1},\theta_i\mid Y_{i,1},\ldots,Y_{i,T-1})
- \mathrm{MSFE}(W^*_{i,T-1},\theta_i\mid Y_{i,1},\ldots,Y_{i,T-1})
\\
&= \sigma_i^2{W}_{i,T-1}^2
+ \kappa_{i,T-1}^2(1-{W}_{i,T-1})^2
- \frac{\sigma_i^2\kappa_{i,T-1}^2}{\kappa_{i,T-1}^2+\sigma_i^2} \\
&= \sigma_i^2\!\left[
{W}_{i,T-1}^2
+ \zeta_{i,T-1}^2(1-{W}_{i,T-1})^2
- \frac{\zeta_{i,T-1}^2}{\zeta_{i,T-1}^2+1}
\right],
\end{split}
\end{align}
where
\begin{align}\label{zetasimple}
\zeta_{i,T-1}^2
:= \frac{\kappa_{i,T-1}^2}{\sigma_i^2}
= \frac{\mathbb{E}\!\left[(A_i-\mu)^2\mid Y_{i,1},\ldots,Y_{i,T-1}\right]}
{\sigma_i^2}.
\end{align}
The form of regret in~\eqref{condregsimple} is similar in spirit to the
regret criterion in statistical decision theory (see, e.g., equation (6)
in \cite{manski2019econometrics}).

\begin{theorem}\label{boundssimple}
Let Assumption~\ref{indep-assumption} hold and assume $\sigma_i^2>0$.
Suppose that $\zeta_{i,T-1}^2$ in~(\ref{zetasimple}) satisfies
$\zeta_{i,T-1}^2\in[0,\tilde{\zeta}_{i,T-1}^2]$, where
$\tilde{\zeta}_{i,T-1}^2>0$. Then maximum regret is
\begin{align*}
&\max_{\zeta_{i,T-1}^2\,\in\,[0,\,\tilde{\zeta}_{i,T-1}^2]}
R^*(W_{i,T-1},\theta_i\mid Y_{i,1},\ldots,Y_{i,T-1}) \\
&= \sigma_i^2
\max\!\left[
{W}_{i,T-1}^2,\;
{W}_{i,T-1}^2 + \tilde{\zeta}_{i,T-1}^2(1-{W}_{i,T-1})^2
- \frac{\tilde{\zeta}_{i,T-1}^2}{\tilde{\zeta}_{i,T-1}^2+1}
\right],
\end{align*}
with $R^*(W_{i,T-1},\theta_i\mid Y_{i,1},\ldots,Y_{i,T-1})$ defined as
in~(\ref{condregsimple}). The weight that minimizes maximum regret over
$W_{i,T-1}\in[0,1]$ is
\begin{align}\label{simplew}
{W}^{IW-MR}_{i,T-1}
= 1 - \frac{1}{\sqrt{\tilde{\zeta}_{i,T-1}^2+1}}.
\end{align}
\end{theorem}

In practice, the value of $\tilde{\zeta}_{i,T-1}^2$ is uncertain, but the
following heuristic rule can be used to obtain feasible weights. Assuming
$T\geq 3$, we define
\begin{align}\label{feasiblezeta}
\widehat{\tilde{\zeta}_{i,T-1}^2}
:= \frac{\max\left\{(Y_{i,1}-\mu)^2,\ldots,(Y_{i,T-1}-\mu)^2\right\}}
{\sum_{t=1}^{T-2}(Y_{i,t}-Y_{i,t+1})^2/[2(T-2)]},
\end{align}
where $\mu$ is either known or approximated by the pooled mean. The
denominator
\[
\sum_{t=1}^{T-2}(Y_{i,t}-Y_{i,t+1})^2/[2(T-2)]
\]
is an unbiased estimator of $\sigma_i^2$, the denominator of
$\zeta_{i,T-1}^2$. The numerator serves as a proxy for an upper bound on
$\kappa_{i,T-1}^2$, the numerator of $\zeta_{i,T-1}^2$.

Although the construction is heuristic, we may interpret our result as a
minimax-regret optimal rule conditional on $\tilde{\zeta}_{i,T-1}^2 =
\widehat{\tilde{\zeta}_{i,T-1}^2}$. 

\subsection{Inverse MSFE Weights (IW-MSFE)}
The inverse MSFE weights (IW-MSFE) do not rely on the model assumptions and
are thus applicable in more general settings. The weights are the same as
those reported in Section~\ref{sec: InverseMSFE}.

\end{appendix}

\nocite{*}

\end{document}